\documentclass[lettersize, journal]{IEEEtran}
\IEEEoverridecommandlockouts

\newtheorem{theorem}{Theorem}
\usepackage[utf8]{inputenc} 
\newtheorem{proposition}{Proposition}

\usepackage{cite}
\usepackage[T1]{fontenc}
\usepackage{amssymb,amsfonts}
\usepackage[cmex10]{amsmath}
\usepackage[linesnumbered, ruled,vlined]{algorithm2e}
\SetAlFnt{\small}
\SetKw{Continue}{continue}
\SetKw{Break}{break}
\let\oldnl\nl
\newcommand{\nonl}{\renewcommand{\nl}{\let\nl\oldnl}}
\usepackage{varwidth}
\ifCLASSOPTIONcompsoc
    \usepackage[caption=false, font=normalsize, labelfont=sf, textfont=sf]{subfig}
\else
\usepackage[caption=false, font=footnotesize]{subfig}
\fi
 
\allowdisplaybreaks
\usepackage{graphicx}
\usepackage{xcolor}
\usepackage{pgfplots}
\usepackage{tikz}
\usepackage{url}
\usetikzlibrary{matrix, positioning, arrows}

\newtheorem{lemma}{Lemma}
\newtheorem{example}{Example}

\def\BibTeX{{\rm B\kern-.05em{\sc i\kern-.025em b}\kern-.08em
    T\kern-.1667em\lower.7ex\hbox{E}\kern-.125emX}}

\begin{document}

\title{On Decoding of Reed-Muller Codes Using a Local Graph Search}

\author{Mikhail~Kamenev%
\thanks{This paper has been presented in part at the 2020 IEEE Information Theory Workshop \cite{shortGS}.}%
\thanks{M. Kamenev is with the Moscow Research Center, Huawei Technologies Co., Ltd., Moscow, Russia. Email: kamenev.mikhail1@huawei.com}}

\maketitle

\begin{abstract}
We present a novel iterative decoding algorithm for Reed-Muller (RM) codes, which takes advantage of a graph representation of the code.
Vertices of the considered graph correspond to codewords, with two vertices being connected by an edge if and only if the Hamming distance between the corresponding codewords equals the minimum distance of the code.
The algorithm uses a greedy local search to find a node optimizing a metric, e.g. the correlation between the received vector and the corresponding codeword.
In addition, the cyclic redundancy check can be used to terminate the search as soon as a valid codeword is found, leading to an improvement in the average computational complexity of the algorithm.
Simulation results for both binary symmetric channel and additive white Gaussian noise channel show that the presented decoder approaches the performance of maximum likelihood decoding for RM codes of length less than 1024 and for the second-order RM codes of length less than 4096.
Moreover, it is demonstrated that the considered decoding approach outperforms state-of-the-art decoding algorithms of RM codes with similar computational complexity for a wide range of block lengths and rates.

\end{abstract}
\begin{IEEEkeywords}
Reed-Muller codes, BSC channels, AWGN channels, near maximum-likelihood decoding, Fast Hadamard Transform.
\end{IEEEkeywords}
\IEEEpeerreviewmaketitle

\section{Introduction}
\IEEEPARstart{B}{inary} Reed-Muller (RM) codes were firstly discovered by Muller \cite{Muller} and then by Reed, who also proposed a majority decoding algorithm for this family of error correction codes \cite{Reed}.
Although it has been proven that RM codes achieve the capacity on an erasure channel under maximum a posteriori (MAP) decoding \cite{Urbanke}, no MAP decoding algorithm that can be efficiently used for decoding of RM codes is known.

However, there are several decoding algorithms of RM codes that allow achieving the performance of a maximum likelihood (ML) decoder for a wide range of block lengths and rates.
For instance, a recursive permutation list decoder with a list size less than or equal to 128 allows achieving near-ML decoding performance for small length $\left( \leq 256 \right)$ codes \cite{Dumer}.
Since this algorithm processes a received codeword sequentially and uses a sorting operation, the latency of a hardware implementation of the recursive permutation list decoder is high.
Permutation-based techniques improving the latency of the recursive permutation list decoder have been proposed in \cite{Stolte, MarvinRMPerm, FHTRMPerm}.
Unfortunately, the complexity of recursive algorithms required for near-ML decoding performance grows exponentially with the code length \cite{Kirill2}.

Another algorithm that performs close to ML decoding is a recursive projection-aggregation (RPA) decoder proposed in \cite{MinYe}.
This algorithm is based on projecting the code on its cosets, recursively decoding the projected codes, and aggregating the reconstructions.
The decoder demonstrates near-ML decoding performance for RM codes of length less than 512 and for the second-order RM codes.
Moreover, it allows for parallel implementation.
Since the code is projected on all its cosets, the RPA algorithm takes a very long time to decode.
Techniques aiming to improve the running time of projection-based decoding have been considered in \cite{KirillRPA, SRPA, RPASubcodes}.

Several other approaches for decoding of RM codes have been proposed recently.
However, they demonstrate near-ML error correction performance for short length codes  \cite{PermGross, RMBP, Pfister, Pfister2, seq1, seq2} or applicable only for a binary erasure channel \cite{Kirill, scInactivation}.

In this paper, we propose a new iterative decoding algorithm for RM codes.
Since recursive decoding needs a large complexity to approach the ML decoding performance for RM codes of length larger than 256 \cite{Kirill2}, the proposed algorithm employs a low-complex recursive decoder \cite{recDec} just to get the initial candidate codeword and then iteratively improves it using a local graph search.
The nodes of the considered graph correspond to codewords.
Any two nodes of the graph are connected by an edge if and only if the Hamming distance between two corresponding codewords equals the minimum distance of the code.
Simulation results demonstrate that the proposed algorithm approaches the performance of the ML decoder for moderate length $\left(\leq 512 \right)$ RM codes and for the second-order RM codes of length less than 4096, with average computational complexity being reasonable.
We also observe that our algorithm outperforms both recursive permutation list decoding and RPA decoding with similar computational complexity in most considered cases.

The rest of the paper is organized as follows.
In Section II, we shortly introduce RM codes.
In Section III, we provide a high-level description of our decoding algorithm.
In Section IV, we present a greedy version of the algorithm that allows to significantly improve the running time.
Numerical results are presented in Section V.
We conclude the paper in Section VI.

\section{RM Codes}
Denote by $f\left(\mathbf{v} \right) = f\left(v_1, \dotsc , v_{m}\right)$ a Boolean function of $m$ variables that is written in the algebraic normal form.
Let $\mathbf{f}$ be the vector of length $2^m$ containing values of $f$ at all of its $2^m$ arguments.
The binary RM code $\mathcal{R}\left(r, m\right)$ of length $n$, $n = 2^m$, and order $r$, $0\leq r \leq m$, is the set of all vectors $\mathbf{f}$, where $f\left(\mathbf{v}  \right)$ is a Boolean function of degree at most $r$.
Note that the minimum distance of the RM code $\mathcal{R}\left(r, m\right)$ equals $2^{m-r}$ \cite[Sec.~13.3]{Sloan}.

RM codes can also be considered in terms of finite geometries.
Denote by $\mathrm{EG}(m,2)$ the Euclidean geometry of dimension $m$ over $\mathrm{GF}(2)$.
By definition, $\mathrm{EG}(m,2)$ contains $2^m$ points, whose coordinates are all binary vectors $\mathbf{v}$ of length $m$.
We associate every subset $\mathcal{S}$ of points of $\mathrm{EG}(m,2)$ with a binary incidence vector of length $2^m$ that contains one in positions $s \in \mathcal{S}$ and zeros elsewhere.
Then the codewords of $\mathcal{R}\left(r, m\right)$ can be considered as incidence vectors of subsets of $\mathrm{EG}(m,2)$.
For instance, the minimum weight codewords of $\mathcal{R}\left(r, m\right)$ are exactly the incidence vectors of $(m-r)$-dimensional subspaces of $\mathrm{EG}(m,2)$  \cite[Sec.~13.4]{Sloan}.

\section{Graph Search Based Decoding of RM Codes}
Consider an RM code $\mathcal{R}\left(r, m\right)$ of dimension $k$ and a graph $\mathcal{G}$ with $2^k$ nodes.
Assign to each node in the graph $\mathcal{G}$ a codeword of $\mathcal{R}\left(r, m\right)$.
Then two nodes are connected by an edge if and only if the Hamming distance between corresponding codewords equals $2^{m - r}$, i.e. the minimum distance of the code.
Assume that $\textbf{y} = \left(y_1, y_2, \dots, y_n \right) $ is a log-likelihood ratio (LLR) vector of a binary-input additive white Gaussian noise (BI-AWGN) channel.
Then, for each node in the graph $\mathcal{G}$, we can assign the correlation metric $M$ \cite[Sec.~10.1]{shulin}, defined by
\begin{equation*}
M = \sum\limits_{i = 1}^{n} \left(1 - 2c_i\right)y_{i},
\end{equation*}
where $\left(c_1, c_2, \dots, c_n \right)$ is a codeword assigned to the node.
Note that the ML decoder returns a codeword with the largest correlation metric $M$ \cite[Sec.~10.1]{shulin}.

Using a graph traversal algorithm (breadth-first search or depth-first search), one can start at a random node, visit all reachable nodes, and return the codeword that corresponds to the node with the largest metric.

\begin{proposition}
The algorithm described above is equivalent to the ML decoder.
\end{proposition}
\begin{IEEEproof}
Observe that the algorithm achieves the performance of the ML decoder if all nodes in the graph $\mathcal{G}$ can be visited.
It is possible if and only if the graph $\mathcal{G}$ is connected.
Since the minimum weight codewords generate the code \cite[Sec.~13.6]{Sloan}, it follows that there is a path between the node associated with the all-zero codeword and any other node of the graph $\mathcal{G}$.
Therefore, $\mathcal{G}$ is a connected graph.
\end{IEEEproof}

Unfortunately, the computational complexity of the algorithm grows exponentially with the code dimension that makes it infeasible for practical usage.
To decrease the complexity of the search algorithm, we propose the following greedy approach.
The algorithm starts at a node associated with an output of low-complex recursive decoding \cite{recDec}.
Then it moves to the adjacent node that has the biggest metric value $M$ and has not been visited before.
The algorithm continues till a maximum number of iterations $N$ is reached or all adjacent nodes are already visited.
The output of the algorithm is the codeword corresponding to the visited node with the biggest metric.

The maximum number of iterations is used here to limit the computational complexity of the algorithm and the amount of memory used to store visited nodes.
In addition, if all adjacent nodes are visited, then the algorithm can be finished until the maximum number of iterations is reached.
Another option that can be used to decrease the average computational complexity is the termination of the algorithm if a codeword that has the best metric found so far satisfies the cyclic redundancy check (CRC).
The formal description of the proposed decoding algorithm is shown in Algorithm \ref{GraphSearchAlg}.

\begin{example}
Consider the RM code $\mathcal{R}\left(2, 3\right)$ and an LLR vector 
\begin{equation}
\mathbf{y} = \left(2.76,  5.68, -6.58,  4.42, -0.09,  3.9 ,  3.56, -1.91 \right).
\label{llrVectorExample}
\end{equation}
 Let us decode this vector using the graph search algorithm with $3$ iterations.
 Assume that the recursive decoder returns a vector $\mathbf{c} = \left(0,0,1,1,1,1,0,0 \right)$ with a metric of $8.44$.
 In the first iteration, the algorithm finds a vector $\mathbf{c_1} = \left(0,0,1,0,1,0,0,0 \right)$ that has the largest correlation ($M = 25.08$) among the codewords at the minimum distance from $\mathbf{c}$. 
In the next iteration, the algorithm performs a search over the codewords at the minimum distance from $\mathbf{c_1}$. The result is the vector  $\mathbf{c_2} = \left(0,0,1,0,0,0,0, 1 \right)$ with a metric of $28.72$. In the last iteration, a vector with the largest metric is $\mathbf{c_1}$.
  Note that this vector has been found before. 
Therefore, the algorithm chooses the vector $\mathbf{c_3} = \left(1,0,1,0,1,0,0,1 \right)$ with the second largest correlation that is equal to $23.38$. 
The output of the algorithm is a codeword with the largest correlation from the set $\left\{\mathbf{c}, \mathbf{c_1}, \mathbf{c_2}, \mathbf{c_3} \right\}$.
 Thus, the output of the algorithm is the codeword $\mathbf{c_2}$.
  It can be readily verified that $\mathbf{c_2}$ is the ML solution for the vector $\mathbf{y}$. 
\end{example}

Note that an algorithm similar to Algorithm \ref{GraphSearchAlg} has been proposed in \cite{lowWeightTrellis} for decoding arbitrary linear codes.
This algorithm starts searching with a codeword returned by a low-complex decoding algorithm and uses decoding over a low-weight trellis instead of the \texttt{NextStep} function.
Although decoding over a low-weight trellis has lower complexity compared to trellis-based ML decoding, it is computationally expensive to run this search procedure many times.
As we will show in the next section, there exists a version of the \texttt{NextStep} function implementing a low-complex greedy search algorithm for RM codes.
It allows running Algorithm \ref{GraphSearchAlg} with a large $N$, with the average computational complexity being reasonable.

 \begin{algorithm}[!t]
 \DontPrintSemicolon
\SetAlgoLined
\KwIn{RM code parameters $r$ and $m$, a vector of LLRs $\mathbf{y}$, a maximum number of iterations $N$, a flag $t$ whether to use the CRC to terminate the algorithm}
\KwOut{A codeword $\mathbf{res}$}
Set $\mathcal{C}$ to be an empty set of codewords\;
Let $\mathbf{c}$ be the result of recursive decoding of the vector $\mathbf{y}$ \;
Add $\mathbf{c}$ to $\mathcal{C}$ \;
$\mathbf{res} \gets \mathbf{c}$ \;
$M \gets 0$ \;
\For {$i = 1, 2, \dots, 2^m$} {
$M \gets M + \left(1 - 2\mathbf{c}[i]\right)\mathbf{y}[i]$ \tcp*{Compute the metric for the codeword $\mathbf{c}$.}
}
\For {$i = 1, 2, \dots, N$}{
$\mathbf{c^\prime}, M^\prime \gets \texttt{NextStep}\left(r, m, \mathbf{y}, \mathbf{c}, \mathcal{C}\right)$ \tcp*{\texttt{NextStep} goes through all codewords $\mathbf{x}$, $\mathbf{x} \notin \mathcal{C}$, such that the Hamming distance between $\mathbf{x}$ and $\mathbf{c}$ equals $2^{m-r}$, and returns the codeword $\mathbf{c^\prime}$ that has the biggest metric. $M^\prime$ is the metric of the output codeword. If all $\mathbf{x} \in \mathcal{C}$, then $M^\prime = -\infty$.}
\If {$M^\prime = -\infty$} {
\Break \tcp*{All adjacent nodes are already visited.}
}
$\mathbf{c} \gets \mathbf{c^\prime}$ \tcp*{Move to an adjacent node.}
Add $\mathbf{c}$ to $\mathcal{C}$ \;
\If {$M^\prime > M$}{ \tcp*{Save the current codeword if it is better than one found so far.}
$\mathbf{res} \gets \mathbf{c}$ \;
$M \gets M^\prime$ \;
\If {$t$ \textnormal{and} $\mathbf{res}$ \textnormal{satisfies the CRC}}{
\Break 
}
}
}
\Return{$\mathbf{res}$}
 \caption{The \texttt{GraphSearch} decoding function} 
 \label{GraphSearchAlg}
 \end{algorithm}

 \section{A Simplified Selection of an Adjacent Node}
The bottleneck of Algorithm \ref{GraphSearchAlg} is the computational complexity of the function \texttt{NextStep}.
A naive implementation of this function goes through all codewords $\mathbf{c}$ of a code $\mathcal{R}\left(r, m\right)$ such that the Hamming distance between $\mathbf{c}$ and the current codeword equals $2^{m-r}$, i.e. the minimum distance of the code.
The number of minimum weight codewords in $\mathcal{R}\left(r, m\right)$ is defined as \cite[Sec.~13.4]{Sloan}
 \begin{equation*}
 A_{2^{m-r}} = 2^r \prod\limits_{i = 0}^{m - r - 1}\frac{2^{m-i} - 1}{2^{m - r - i} - 1}.
 \end{equation*}
 For instance, the number of minimum weight codewords of the code $\mathcal{R}\left(4, 9\right)$ is approximately $53 \cdot 10^6$.
Thus, a brute-force algorithm that goes through all adjacent nodes takes a long time to find the next codeword.
To solve this issue, we propose a greedy search algorithm that has the best-case running time of $\mathcal{O}\left(n\log \left(\max \left\lbrace N, n \right\rbrace \right) \right)$.

Consider an RM code $\mathcal{R}\left(r, m\right)$ of length $n$ and a rooted tree $\mathcal{T}$ constructed in the following way.
The root of the tree is assigned $\mathrm{EG}(m,2)$.
Each node at depth $i$, $i \in \left\lbrace 0, \dots, r - 1 \right\rbrace$, has $2^{m - i + 1} - 2$ children containing all possible $(m - i - 1)$-dimensional subspaces of the parent node \cite[Appx. B]{Sloan}.
By the construction of the tree $\mathcal{T}$, the leaf nodes contain all $(m - r)$-dimensional subspaces of $\mathrm{EG}(m,2)$.
Recall that the union of these subspaces contains all minimum weight codewords of $\mathcal{R}\left(r, m\right)$  \cite[Sec.~13.4]{Sloan}.
Thus, if we replace the points of subspaces with the corresponding coordinate positions in the codeword, then the leaves of the tree $\mathcal{T}$ will store non-zero positions of all minimum weight codewords in $\mathcal{R}\left(r, m\right)$.
The tree $\mathcal{T}$ for the RM code $\mathcal{R}\left(2, 3\right)$ is illustrated in Fig. \ref{MinCodewordsTree}.
Observe that different leaf nodes can contain non-zero positions of the same codeword.

The tree $\mathcal{T}$ can be easily constructed using codewords of the first order RM codes. 
The root of the tree is assigned a set $\left[n \right] \triangleq \left\{1, 2, \dots, n \right\}$.
Let $\mathbf{c} = \left(c_1, c_2, \dots, c_n\right)$ be a codeword of weight $n/2$ in $\mathcal{R}\left(1, m \right)$ and let $\mathcal{S}\left(\mathbf{c} \right) = \left\{ i \in \left[n \right] \,\mid\, c_i = 1 \right\}$ be a set of non-zero positions in $\mathbf{c}$.
Since the codewords of weight $n/2$ in $\mathcal{R}\left(1, m \right)$ are incidence vectors of $(m - 1)$-dimensional subspaces  \cite[Sec.~13.4]{Sloan}, the children nodes of the root contain sets $\mathcal{S}\left(\mathbf{c_i} \right)$, $i \in \left[2n - 2 \right]$, where $\mathbf{c_1}, \dots, \mathbf{c_{2n - 2}}$ are distinct codewords of weight $n/2$ in $\mathcal{R}\left(1, m \right)$.
Let us consider a child of the root with a set $\mathcal{S}\left(\mathbf{c} \right) = \left\{s_1, s_2, \dots, s_{n/2} \right\}$ and let us consider the code that is obtained by taking those codewords of the original RM code that are zero on $\left[n \right] \setminus \mathcal{S}\left(\mathbf{c} \right)$ and deleting these positions.
Since the set $\mathcal{S}\left(\mathbf{c} \right)$ defines an $(m - 1)$-dimensional subspace, the shortened code is an RM code, whose codewords are incidence vectors of subsets in $\mathrm{EG}(m-1,2)$ \cite[Sec.~13.4]{Sloan}.
Therefore, we can apply a similar procedure to generate the children of the considered node as for the root node.
Let $\mathbf{\hat{c}} = \left(\hat{c}_1, \hat{c}_2, \dots, \hat{c}_{n/2}\right)$ be a codeword of weight $n/4$ in $\mathcal{R}\left(1, m - 1 \right)$ and let $\mathcal{\hat{S}}\left(\mathbf{\hat{c}} \right) = \left\{s_i \,\mid\, i \in \left[n/2 \right], \hat{c}_i = 1 \right\}$.
Then the children nodes of the node with the set  $\mathcal{S}\left(\mathbf{c} \right)$ contain the sets $\mathcal{\hat{S}}\left(\mathbf{\hat{c}_i} \right)$, $i \in \left[n - 2 \right]$, where $\mathbf{\hat{c}_1}, \dots, \mathbf{\hat{c}_{n - 2}}$ are distinct codewords of weight $n/4$ in $\mathcal{R}\left(1, m - 1\right)$.
This process continues until the node with a set of size $2^{m-r}$, which contains non-zero positions of a minimum weight codeword in $\mathcal{R}\left(r, m\right)$, is reached.

\begin{figure}
 \centering
 \begin{tikzpicture}[ font=\footnotesize,
 level 1/.style={level distance = 2.5cm, sibling distance = 4cm, font=\footnotesize},
 level 2/.style={level distance = 2.5cm, sibling distance = 1.1cm, font=\footnotesize}, scale=0.75]
\node [circle,draw] (z){$\left\lbrace1,2, \dots, 8 \right\rbrace$}
  child {node [circle,draw] (a1) {$\left\lbrace1,3,5,6 \right\rbrace$}
        child {node [circle,draw] (c1) {$\left\lbrace1,5 \right\rbrace$}}
     child {node {$\ldots$}  edge from parent[draw=none]}
    child {node [circle,draw] (c3) {$\left\lbrace3,5 \right\rbrace$}}
  }
  child {node [circle,draw] (a3) {$\left\lbrace1,4,5,8 \right\rbrace$}
        child {node [circle,draw] (c1) {$\left\lbrace1,5 \right\rbrace$}}
     child {node {$\ldots$}  edge from parent[draw=none]}
     child {node [circle,draw] (c6) {$\left\lbrace4,5 \right\rbrace$}}
  }
  child {node [circle,draw] (a4) {$\left\lbrace2,3,5,8 \right\rbrace$}  
        child {node [circle,draw] (c1) {$\left\lbrace2,5 \right\rbrace$}}
     child {node {$\ldots$}  edge from parent[draw=none]}
     child {node [circle,draw] (c6) {$\left\lbrace3,5 \right\rbrace$}}
  };
  \path (a3) -- (a4) node [midway, transform canvas={xshift=0pt, yshift=0pt}] {$\ldots$};
  \end{tikzpicture}
\caption{A tree that is used to enumerate all minimum weight codewords of the RM code $\mathcal{R}\left(2, 3\right)$} 
 \label{MinCodewordsTree}
 \end{figure}
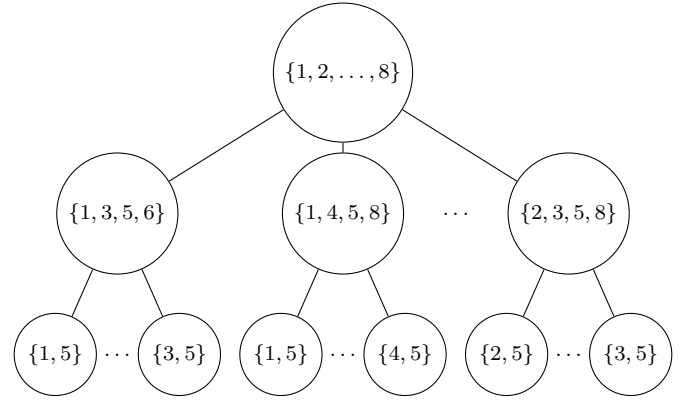


The greedy search algorithm can be implemented using the tree $\mathcal{T}$.
Recall that this function aims to find a codeword at the minimum distance from the current codeword $\mathbf{c}$ that has the best metric.
To decrease the complexity of the algorithm, we consider a depth-first search that is terminated when the first leaf is reached.
Thus, it is required to find a leaf in the tree $\mathcal{T}$ such that if we flip bits in the codeword $\mathbf{c}$ at the positions in the set of this leaf, we will get the codeword with the best metric.

Consider a leaf with a set $\mathcal{S}$ and assume that $\mathbf{c} = \left(c_1, c_2, \dots, c_{n} \right)$ is the current codeword.
If we choose this leaf as a solution, then the metric of the output codeword is calculated as 
\begin{equation}
\sum\limits_{i = 1}^{n} \left(1 - 2c_i\right) y_i - 2\sum\limits_{i \in \mathcal{S}} \left(1 - 2c_i\right) y_i,
\label{OutputMetric}
\end{equation}
where $\left(y_1, y_2, \dots, y_{n} \right)$ is a vector of input LLRs.
Since we are required to maximize the metric of the output codeword, the algorithm needs to find a leaf that minimizes the second sum of \eqref{OutputMetric}.

Consider a node of the tree $\mathcal{T}$ that contains a set $\mathcal{S} = \left\lbrace s_1, s_2, \dots, s_{\bar{n}}\right\rbrace$, $1 \leq s_1 < s_2 < \dots < s_{\bar{n}} \leq n$.
Assume that this node is not a leaf.
Consequently, it has $2\bar{n} - 2$ children nodes.
Since we consider the greedy depth-first search, the algorithm needs to move to a child node that seems best at the moment.
Here, we propose to use the following heuristic: For every child node with a set $\mathcal{\bar{S}}$, we compute a metric
\begin{equation}
2\sum\limits_{i \in \mathcal{\bar{S}}}\left(1 - 2c_i\right)y_i
 \label{TreeHeuristic}
 \end{equation}
and move to the node with the smallest metric value. 
The algorithm continues until the leaf node is reached.
The output of the algorithm is the codeword that is obtained from the current one by flipping bits at the positions in the set of the leaf node.
Observe that the leaf with minimum metric value \eqref{TreeHeuristic} provides the best solution in terms of \eqref{OutputMetric}. 

A straightforward calculation of \eqref{TreeHeuristic} for all children of a node with a set $\mathcal{S}$ takes $\mathcal{O}\left(\bar{n}^2 \right)$ time.
Since a node of the tree $\mathcal{T}$ contains a set with non-zero positions of a codeword of the first order RM code, we can improve the running time of a naive algorithm by taking advantage of the Fast Hadamard Transform (FHT) \cite{Green}.
Let $\mathbf{\bar{y}} = \left(\bar{y}_1, \bar{y}_2, \dots, \bar{y}_{\bar{n}} \right)$ be a vector such that $\bar{y}_i = \left( 1 - 2 c_{s_i}\right) y_{s_i}$ and let 
\begin{equation} 
\mathbf{H} = \begin{bmatrix}
    1 &  1 \\
    1 & -1 \\
\end{bmatrix}^{\otimes \log_2 \bar{n}},
\label{firstOrderRMCodewords}
\end{equation}
where $\mathbf{X}^{\otimes t}$ denotes $t$-times Kronecker product of the matrix $\mathbf{X}$ with itself.
Then the result of the FHT applied to the vector $\mathbf{\bar{y}}$ can be expressed in terms of matrix multiplication as $\left(h_1, h_2, \dots, h_{\bar{n}} \right) = \mathbf{\bar{y}} \mathbf{H}$.
Note that the columns of the matrix $\mathbf{H}$ correspond to half of the first-order RM code codewords.
Namely, if $\left(x_1, x_2, \dots, x_{\bar{n}} \right)^T$ is a column of the matrix $\mathbf{H}$ associated with the codeword $\mathbf{c}$, then $c_i = (1 - x_i)/2$.
The rest codewords are defined by the columns of $-\mathbf{H}$.
Therefore, the values of \eqref{TreeHeuristic} for the children nodes can be calculated as $\left(-1\right)^ah_i + h_{1}$, where $i \in \left\lbrace2, \dots \bar{n}\right\rbrace$, $a \in \left\lbrace0, 1 \right\rbrace$.
The running time of this approach is $\mathcal{O}\left(\bar{n} \log \bar{n} \right)$ \cite{Green}.

\begin{example}
Consider a similar scenario as in Example 1.
Namely, the RM code $\mathcal{R}\left(2, 3\right)$ and LLR vector \eqref{llrVectorExample}.
Assume that $\mathbf{c} = \left(0, 0, 1, 1, 1, 1, 0, 0 \right)$ is the current codeword.
The correlation metric for this codeword equals $8.44$.
Let us find a codeword at the minimum distance from $\mathbf{c}$ that has the largest correlation using the depth-first algorithm described above. 
In the first step, the algorithm computes the metric \eqref{TreeHeuristic} for each child of the root and moves to the node with the set $\mathcal{S} = \left\{2, 4, 6, 8 \right\}$ that has the smallest metric, namely, $-9.1$.
Then, the algorithm applies the same procedure to the node with the set $\mathcal{S}$ and selects a node with the set $\left\{4,6 \right\}$.
The metric value for this node equals $-16.64$ and, from \eqref{OutputMetric}, it follows that if we flip bits at positions 4 and 6, we will get a codeword $\mathbf{c^\prime} = \left(0, 0, 1, 0, 1, 0, 0, 0 \right)$ with a metric of $25.04$.
Note that $\mathbf{c^\prime}$ and $\mathbf{c_1}$ found in Example 1 using an optimal algorithm are the same.
\end{example}

The version of the greedy algorithm described above is able to find only one solution.
If this solution is in the set $\mathcal{C}$ (see line 3 of Algorithm \ref{GraphSearchAlg}), then the break statement in line 12 of Algorithm \ref{GraphSearchAlg} will terminate the for-loop in lines 9 -- 23.
We observe that in this case, the algorithm terminates before the correct codeword is found, significantly decreasing the error-rate performance of the decoder.
Therefore, one requires to modify the greedy approach described above in such a way that it can visit multiple leaves and select one that generates the codeword $\mathbf{c}$, $\mathbf{c} \notin \mathcal{C}$.
For instance, one can choose $l$ children of the root of the tree $\mathcal{T}$ and run the depth-first search for each child independently.
Then the algorithm can generate up to $l$ different codewords and choose one that is not in $\mathcal{C}$.
We call parameter $l$ a search breadth of the decoding algorithm.

Note that the complexity of the decoding algorithm depends on the value of the search breadth.
If the search breadth is too small, then the algorithm has a high probability of termination before reaching the correct codeword.
On the other hand, a large search breadth significantly increases the running time of the algorithm.
We found that a good trade-off here is to run the search with a relatively small $l$, for instance, $l = 8$.
However, if all generated codewords are in the set $\mathcal{C}$, then we allow running the depth-first search for $\bar{l}$ additional children of the root of the tree $\mathcal{T}$.
This operation is allowed to be run only $s$ times for the one launch of the decoder.
If the algorithm is not able to find a codeword $\mathbf{c}$, $\mathbf{c} \notin \mathcal{C}$, for the given search breadth $l$ and the algorithm has already run the depth-first search for  $\bar{l}$ additional children $s$ times, then decoding is terminated.
This limitation is introduced to decrease the average computational complexity of the algorithm.
Note that if the search for each child node of the root runs sequentially, one can stop the search procedure as soon as a codeword $\mathbf{c}$, $\mathbf{c} \notin \mathcal{C}$, that improves the metric of the current best codeword is found.

\begin{example}
Let us consider the RM code $\mathcal{R}\left(2, 3\right)$ and let us decode LLR vector \eqref{llrVectorExample} using a greedy approach discussed above.
Let $N = 10$, $l = 2$, $\bar{l} = 1$, $s = 1$ and assume that $\mathbf{c} = \left(0, 0, 1, 1, 1, 1, 0, 0 \right)$ is the output of the recursive decoder.
As in Example 2, in the first iteration, the algorithm finds a vector $\mathbf{c_1} = \left(0, 0, 1, 0, 1, 0, 0, 0 \right)$ with a metric of 25.08 by applying the depth-first search to the child node with the best metric \eqref{TreeHeuristic}.
Since the metric of $\mathbf{c_1}$ is greater than the metric of $\mathbf{c}$, the algorithm does not run the depth-first search for the second-best child and moves to the next iteration. 
In the second iteration, the best child of the root has the set $\left\{1,4,5,8 \right\}$ and the depth-first search for this node returns the set $\left\{5,8 \right\}$.
This set generates the codeword $\mathbf{c_2} = \left(0, 0, 1, 0, 0, 0, 0, 1 \right)$ with a metric of 28.72. 
As in the first iteration, the depth-first search for the best child results in a new codeword with a better metric than that of the codewords $\mathbf{c}$ and $\mathbf{c_1}$.
Therefore, the algorithm terminates the second iteration and moves to the next one.
In the third iteration, the two best children of the root are defined by sets $\left\{1,4,5,8 \right\}$ and $\left\{5,6,7,8 \right\}$.
The depth-first search for both these nodes results in the set $\left\{5,8 \right\}$ generating the codeword $\mathbf{c_1}$ that has been already found in the first iteration.
Since the search for the $l$ best children does not result in a new codeword, the algorithm runs the depth-first search for the third-best child defined by the set $\left\{1, 2,5,6 \right\}$.
This search results in the set $\left\{1,5\right\}$ and the codeword $\mathbf{c_3} = \left(1, 0, 1, 0, 1, 0, 0, 1 \right)$ with a metric of 23.38. 
In the fourth iteration, as in the previous one, the search for the two best children of the root defined by sets $\left\{1,4,5,8 \right\}$ and $\left\{1,2,5,6 \right\}$ results in the codeword $\mathbf{c_2}$ that has been already found.
Since the search has been already performed for additional $\bar{l} $ nodes $s$ times, 
the algorithm terminates here and returns a codeword with the best metric from the set $\left\{\mathbf{c}, \mathbf{c_1}, \mathbf{c_2}, \mathbf{c_3} \right\}$.
Note that the result in each iteration is the same as in Example 1 that considers an optimal algorithm to find the next codeword.
\end{example}

The formal description of the proposed greedy search approach is shown in Algorithm \ref{NextStepGreedy}. 
This algorithm uses the \texttt{NextStepGreedyRec} function presented in Algorithm \ref{NextStepGreedyRec}.
The \linebreak \texttt{NextStepGreedyRec} function performs a depth-first search to find a leaf that minimizes metric \eqref{TreeHeuristic}.
At each recursive step, the function \texttt{NextStepGreedyRec} chooses a node of the tree $\mathcal{T}$ that minimizes  \eqref{TreeHeuristic} (see lines 12 -- 22).
Condition in line 1 checks whether a leaf node has been reached.
If so, then the algorithm flips bits of the current codeword $\mathbf{c}$ in the coordinates from the vector $\mathbf{v}$ (see lines 3 -- 5).
Note that, for simplicity, we use vectors instead of sets to store positions of bits to be flipped.
If the new codeword is not in the set $\mathcal{C}$, then the new codeword and the value of \eqref{TreeHeuristic} are returned.
Otherwise, the algorithm returns $\infty$ instead of the value of metric \eqref{TreeHeuristic}.

 \begin{algorithm}[t]
 \DontPrintSemicolon
\SetAlgoLined
\KwIn{RM code parameters $r$ and $m$, a vector of LLRs $\mathbf{y}$, a codeword $\mathbf{c}$, a set of codewords $\mathcal{C}$, search breadth $l$ and $\bar{l}$}
\KwOut{A codeword $\mathbf{res}$, a metric $M$ of the codeword $\mathbf{res}$, a flag $ec$ whether the algorithm did not find a codeword using $l$ attempts}

$\mathbf{res} \gets \mathbf{c}$, $M_{\mathbf{c}} \gets 0$, $M \gets \infty$, $ec \gets 0$ \;
\For {$i = 1, 2, \dots, 2^m$} {
$M_{\mathbf{c}} \gets M_{\mathbf{c}} + \left(1 - 2\mathbf{c}[i]\right)\mathbf{y}[i]$ \;
}
Let $\mathbf{c}^i, 1 \leq i \leq 2^{m+1} - 2$, be all codewords of the code $\mathcal{R}\left(1, m\right)$ with the Hamming weight $2^{m-1}$\;
Let $\mathbf{a}$ be a vector of length $2^{m+1} - 2$, $\mathbf{a}[i] \gets 2 \cdot \sum\limits_{j, \mathbf{c}^i[j] = 1}\left(1 - 2\mathbf{c}\left[j\right]\right) \mathbf{y}[j]$ \;
Let $\mathbf{k}$ be a vector that index the vector of metrics $\mathbf{a}$ in ascending order \;
\For {$i = 1, 2, \dots, l + \bar{l}$}{
\If {$i > l$}{
\uIf {$M = \infty$}{
$ec \gets 1$ \;
}\Else {
\Break
}
}
Let $\mathbf{v}$ be a vector of size $2^{m-1}$ \;
$h \gets 1$ \;
\For {$j = 1, 2, \dots, 2^m$}{
\lIf {$\mathbf{c}^{\mathbf{k}[i]}[j] = 1$}{
$\mathbf{v}[h] \gets j$, $h \gets h + 1$ 
}
}
$\mathbf{c}^\prime, M^\prime \gets \texttt{NextStepGreedyRec}(r-1, m-1, \mathbf{y}, \mathbf{c}, \mathcal{C}, \mathbf{v}, \mathbf{a}\left[\mathbf{k}[i]\right])$  \;
\lIf {$M^\prime < M$}{ $M \gets M^\prime$, $\mathbf{res} \gets \mathbf{c}^\prime$}
\lIf {$M < 0$}{ {\Break}
}
}
\uIf {$M = \infty$}{
$M \gets -\infty$ \;
}\Else{
 $M \gets M_\mathbf{c} - M^\prime$\;
 }
\Return $\mathbf{res}, M, ec$\;
 \caption{The \texttt{NextStepGreedy} function} 
 \label{NextStepGreedy}
 \end{algorithm} 
 
 \begin{algorithm}[!t]
 \DontPrintSemicolon
\SetAlgoLined
\KwIn{RM code parameters $r$ and $m$, a vector of LLRs $\mathbf{y}$, a codeword $\mathbf{c}$, a set of codewords $\mathcal{C}$, a vector of coordinates $\mathbf{v}$, a local metric value $M^\prime$}
\KwOut{A codeword $\mathbf{res}$, a metric $M$ that is used to find $\mathbf{res}$}
\If {$r = 0$}{
 $\mathbf{res} \gets \mathbf{c}$ \;
\For {$i = 1, 2, \dots, 2^m$}{
 $\mathbf{res}[\mathbf{v}[i]] \gets \mathbf{res}[\mathbf{v}[i]] \oplus 1$ \;
}
\uIf {$\mathbf{res} \in \mathcal{C}$}{
\Return $\mathbf{res}, \infty$ \;
}\Else {
\Return $\mathbf{res}, M^\prime$ \;
}
}

Let $\mathbf{c}^i, 1 \leq i \leq 2^{m+1} - 2$, be all codewords of the code $\mathcal{R}\left(1, m\right)$ with the Hamming weight $2^{m-1}$ \;
Let $\mathbf{a}$ be a vector of length $2^{m+1} - 2$, $\mathbf{a}[i] \gets 2 \cdot \sum\limits_{j, \mathbf{c}^i[j] = 1}\left(1 - 2\mathbf{c}\left[\mathbf{v}\left[j\right]\right]\right)\mathbf{y}[\mathbf{v}\left[j\right]]$ \;
Let $k$ be an index of the minimum value in $\mathbf{a}$ \;

Let $\mathbf{\hat{v}}$ be a vector of size $2^{m-1}$ \;
$h \gets 1$ \;
\For {$j = 1, 2, \dots, 2^m$}{
\If {$\mathbf{c}^{k}[j] = 1$}{
$\mathbf{\hat{v}}[h] \gets \mathbf{v}[j]$, $h \gets h + 1$ \;
}
}

\Return $\texttt{NextStepGreedyRec}(r - 1, m - 1, \mathbf{y}, \mathbf{c}, \mathcal{C}, \mathbf{\hat{v}}, \mathbf{a}[k])$ 

 \caption{The \texttt{NextStepGreedyRec} function} 
 \label{NextStepGreedyRec}
 \end{algorithm} 

The \texttt{NextStepGreedy} function calculates metric \eqref{TreeHeuristic} for the children of the root node of the tree $\mathcal{T}$ (see lines 5 -- 7).
Then it uses the \texttt{NextStepGreedyRec} function to perform a depth-first search for $l$ children with the smallest metric values.
Note that $l + \bar{l}$ must be less than or equal to $2^{m+1} - 2$.
If a new codeword has been found during the first $l$ iterations, then the loop in lines 8 -- 24 is terminated (see lines 9 -- 15).
If additional $\bar{l}$ nodes are required to find a new codeword, then the algorithm sets a flag $ec$ to one.
This flag is used to count how many times the algorithm used the option to run the depth-first search for additional $\bar{l}$ nodes.
If the \texttt{NextStepGreedy} function returns $ec$ that equals one $s$ times during decoding of the received vector, then the graph search algorithm prohibits the \texttt{NextStepGreedy} function from visiting $\bar{l}$ extra nodes.
For instance, it can be done by setting $\bar{l}$ to be equal to zero.
Note that the loop in lines 8 -- 24 is also terminated if a codeword with a better metric has been found (see line 23).
If the \texttt{NextStepGreedy} function is not able to find a codeword that is not in $\mathcal{C}$, then the function returns the input codeword $\mathbf{c}$ and $-\infty$ as a metric value.
The formal description of the graph search algorithm that uses the \texttt{NextStepGreedy} function is presented in Algorithm \ref{GreedyGraphSearchAlg}.

\begin{algorithm}[!t]
 \DontPrintSemicolon
\SetAlgoLined
\KwIn{RM code parameters $r$ and $m$, a vector of LLRs $\mathbf{y}$, a maximum number of iterations $N$, a flag $t$ whether to use the CRC to terminate the algorithm, search breadth $l$ and $\bar{l}$, a number of extra $\bar{l}$ attempts to find a codeword $s$}
\KwOut{A codeword $\mathbf{res}$}
Set $\mathcal{C}$ to be an empty set of codewords\;
Let $\mathbf{c}$ be the result of recursive decoding of the vector $\mathbf{y}$ \;
Add $\mathbf{c}$ to $\mathcal{C}$ \;
$\mathbf{res} \gets \mathbf{c}$ \;
$M \gets 0$ \;
\For {$i = 1, 2, \dots, 2^m$} {
$M \gets M + \left(1 - 2\mathbf{c}[i]\right)\mathbf{y}[i]$ \tcp*{Compute the metric for the codeword $\mathbf{c}$.}
}
\For {$i = 1, 2, \dots, N$}{
$\mathbf{c^\prime}, M^\prime, ec \gets \texttt{NextStepGreedy}\left(r, m, \mathbf{y}, \mathbf{c}, \mathcal{C}, l, \bar{l}\right)$ \;
\lIf {$ec = 1$} {$s \gets s - 1$}
\lIf {$s = 0$} {$\bar{l} \gets 0$}
\If {$M^\prime = -\infty$} {
\Break \tcp*{A codeword $\mathbf{c^\prime}$ such that $\mathbf{c^\prime} \notin \mathcal{C}$ is not found.}
}
$\mathbf{c} \gets \mathbf{c^\prime}$ \tcp*{Move to an adjacent node.}
Add $\mathbf{c}$ to $\mathcal{C}$ \;
\If {$M^\prime > M$}{ \tcp*{Save the current codeword if it is better than one found so far.}
$\mathbf{res} \gets \mathbf{c}$ \;
$M \gets M^\prime$ \;
\If {$t$ \textnormal{and} $\mathbf{res}$ \textnormal{satisfies the CRC}}{
\Break 
}
}
}
\Return{$\mathbf{res}$}
 \caption{A greedy version of the graph search decoding algorithm} 
 \label{GreedyGraphSearchAlg}
 \end{algorithm} 
 
We now consider the running time and the space complexity of Algorithm \ref{GreedyGraphSearchAlg}. 
First, we consider the computational complexity and space requirements of the functions \texttt{NextStepGreedyRec} and \texttt{NextStepGreedy}.
\begin{lemma}
The running time of the \texttt{NextStepGreedyRec} function is \linebreak $\mathcal{O}\left(n \log \left(\max \left\lbrace N, n \right\rbrace \right) \right)$, $n = 2^m$. 
\end{lemma}
\begin{IEEEproof}
Consider the base case in lines 2 -- 10 of Algorithm \ref{NextStepGreedyRec}.
The algorithm flips at most $2^{m-r}$ bits in line 4 and checks whether a new codeword is in the set $\mathcal{C}$.
This check can be implemented efficiently using a red-black tree guaranteeing that the worst-case running time of the search operation is $\mathcal{O}\left(\log N \right)$ \cite[Sec.~13]{Cormen}.
Since the per-bit comparison of two codewords takes linear time, the base case running time is $\mathcal{O}\left(n \log N \right)$.

The elements of the array $\mathbf{a}$ (see line 13) can be calculated using the FHT that takes time $\mathcal{O}\left(n \log n \right)$.
Since the algorithm computes the array $\mathbf{a}$ at each recursive step, the total running time of this operation is
\begin{equation}
\mathcal{O}\left(\sum\limits_{i = 0}^{r - 1}2^{m - i}\log{}\left(2^{m - i}\right)\right) = \mathcal{O}\left(n \log n \right).
\end{equation}

Observe that the total running time of the for loop in lines 17 -- 21 and the minimum value search in line 14 is $\mathcal{O}\left(n\right)$.
Thus, the \texttt{NextStepGreedyRec} function takes \linebreak $\mathcal{O}\left(n \log\left(\max \left\lbrace N, n \right\rbrace \right) \right)$ time.
\end{IEEEproof}
\begin{lemma}
The space complexity of the \texttt{NextStepGreedyRec} function is $\mathcal{O}\left(n^2 \right)$, $n = 2^m$. 
\end{lemma}
\begin{IEEEproof}
Observe that the base case and lines 14 -- 21 of Algorithm \ref{NextStepGreedyRec} use linear space. 
The algorithm also stores all codewords of $\mathcal{R}\left(1, h\right)$, $h \in \left\lbrace m - r + 1, m - r + 2, \dots, m\right\rbrace$, with the Hamming weight $2^{h-1}$ to select the codeword with the smallest metric efficiently (line 14 and line 18).
Since the codewords of the first order RM codes can be defined recursively using \eqref{firstOrderRMCodewords}, it follows that it is enough to store only codewords of $\mathcal{R}\left(1, m\right)$ that takes $\mathcal{O}\left(n^2 \right)$ space.
Therefore, the algorithm uses $\mathcal{O}\left(n^2 \right)$ space.
\end{IEEEproof}
\begin{lemma} 
The best-case time of the \texttt{NextStepGreedy} function is $\mathcal{O}\left(n \log \left(\max \left\lbrace N, n \right\rbrace \right) \right)$. The worst-case time of the \texttt{NextStepGreedy} function is $\mathcal{O}\left((l + \bar{l}) \cdot n \log \left(\max \left\lbrace N, n \right\rbrace \right) \right)$.
\end{lemma}
\begin{IEEEproof}
Consider line 6 and line 7 of Algorithm \ref{NextStepGreedy}. 
As in Algorithm \ref{NextStepGreedyRec}, the vector $\mathbf{a}$ is computed in $\mathcal{O}\left(n \log n \right)$. 
Sorting in line 7 also takes $\mathcal{O}\left(n \log n \right)$ \cite[Part~II]{Cormen}. 
Observe that all the rest calculations require linear time exclusive of line 21.

In line 21, the algorithm uses the \texttt{NextStepGreedyRec} function to find a new codeword.
In the best case, this function is called only once.
Since this function runs in $\mathcal{O}\left(n \log \left( \max \left\lbrace N, n \right\rbrace \right) \right)$, the best-case running time of the \texttt{NextStepGreedy} function is $\mathcal{O}\left(n \log  \left( \max \left\lbrace N, n \right\rbrace \right) \right)$.
In the worst case, the \texttt{NextStepGreedyRec} function is called $l + \bar{l}$ times.
Thus, the worst-case running time is $\mathcal{O}\left(\left(l + \bar{l}\right) \cdot n \log  \left( \max \left\lbrace N, n \right\rbrace \right) \right)$.
\end{IEEEproof}
\begin{lemma}
The space complexity of the \texttt{NextStepGreedy} function is $\mathcal{O}\left(n^2 \right)$. 
\end{lemma}
\begin{IEEEproof}
As the \texttt{NextStepGreedyRec} function, the \texttt{NextStepGreedy} function stores codewords of $\mathcal{R}\left(1, m\right)$ that take $\mathcal{O}\left(n^2 \right)$ space.
Observe that the auxiliary variables use linear space.
Therefore, the space requirement of the \texttt{NextStepGreedy} function is indeed $\mathcal{O}\left(n^2 \right)$.
\end{IEEEproof}
 
Note that Algorithm \ref{NextStepGreedy} allows for parallel implementation.
One can run the first $l$ iterations of the for loop in lines 8 -- 24 of the algorithm in parallel and select a codeword with the best metric.
If all codewords returned by $l$ calls of the \texttt{NextStepGreedyRec} function are in the set $\mathcal{C}$, then the algorithm can run $\bar{l}$ more iterations in parallel.
Consequently, the best-case running time of the parallel implementation of the \texttt{NextStepGreedy} function is $\mathcal{O}\left(l \cdot n \log  \left( \max \left\lbrace N, n \right\rbrace \right) \right)$, while the worst-case time is still $\mathcal{O}\left(\left(l + \bar{l}\right) \cdot n \log  \left( \max \left\lbrace N, n \right\rbrace \right) \right)$.
The space complexity of the parallel implementation of the \texttt{NextStepGreedy} function is $\mathcal{O}\left(n^2 \right)$.
Indeed, the functions \texttt{NextStepGreedy} and \texttt{NextStepGreedyRec} only read information from the memory that stores codewords of $\mathcal{R}\left(1, h\right)$, $h \in \left\lbrace m - r + 1, m - r + 2, \dots, m\right\rbrace$.
Thus, only one instance of these codewords can be stored in memory.
Moreover, since $l + \bar{l} \leq 2^{m+1} - 2$, the auxiliary variables take $\mathcal{O}\left(n^2 \right)$ space.
Consequently, the space requirement of the sequential and the parallel versions of the algorithm is the same. 

\begin{theorem}
The worst-case running time of Algorithm \ref{GreedyGraphSearchAlg} is $\mathcal{O}\left(L \cdot n \log  \left( \max \left\lbrace N, n \right\rbrace \right) \right)$, where $L = N \cdot l + s \cdot \bar{l}$.
\end{theorem}
\begin{IEEEproof}
Let us consider the running time of the main loop in lines 9 -- 25 of Algorithm \ref{GreedyGraphSearchAlg}.
Recall that we assume that the set $\mathcal{C}$ is implemented using a red-black tree.
Therefore, the insertion operation in line 17 of Algorithm \ref{GreedyGraphSearchAlg} takes $\mathcal{O}\left(n \log N \right)$ time \cite[Sec.~13]{Cormen}.
In addition, we assume that a reasonable length CRC is used.
Under this assumption, we can conclude that line 21 of Algorithm \ref{GreedyGraphSearchAlg} takes a shorter running time than line 10.
In the worst case, the break statement does not terminate the for loop in lines 9 -- 25 of Algorithm \ref{GreedyGraphSearchAlg} and the algorithm uses the possibility to execute $\bar{l}$ additional iterations of the for loop in lines 8 -- 24 of Algorithm \ref{NextStepGreedy} $s$ times.
Thus, the worst-case time of the main loop is $\mathcal{O}\left(\left(N \cdot l + s \cdot \bar{l} \right) n \log  \left( \max \left\lbrace N, n \right\rbrace \right) \right)$.

Since recursive decoding takes $\mathcal{O}\left(n \log n\right)$ time \cite{recDec}, Algorithm \ref{GreedyGraphSearchAlg} takes \linebreak $\mathcal{O}\left(L \cdot n \log  \left( \max \left\lbrace N, n \right\rbrace \right) \right)$ worst-case time.
\end{IEEEproof}
 \begin{theorem}
Algorithm \ref{GreedyGraphSearchAlg} uses $\mathcal{O}\left(\max \left\lbrace N, n \right\rbrace \cdot n\right)$ space.
\end{theorem}
\begin{IEEEproof}
On the one hand, the algorithm needs to store the set of codewords $\mathcal{C}$.
Assume that this set is implemented using a red-black tree.
Thus, it takes $\mathcal{O}\left( N \cdot n\right)$ space \cite[Sec.~13]{Cormen}.
On the other hand, the algorithm calls the \texttt{NextStepGreedy} function that uses $\mathcal{O}\left( n^2\right)$ space.
Since the auxiliary variables use $\mathcal{O}\left(1\right)$ space and recursive decoding uses linear space, the space complexity of Algorithm \ref{GreedyGraphSearchAlg} is $\mathcal{O}\left(\max \left\lbrace N, n \right\rbrace \cdot n\right)$.
\end{IEEEproof}

\section{Simulation results}
We consider RM codes of length 256 and 512 and the second-order RM codes of length 1024 and 2048.
The graph search decoder (referred to as GS) is compared to RPA decoding without list \cite{MinYe} and the recursive permutation list decoder (referred to as RPL) \cite{Dumer}.
We assume that the graph search decoder uses the sequential implementation of Algorithm \ref{NextStepGreedy} and runs with parameters $l = 8$, $\bar{l} = 8$, and $s = 5$.
Note that the simulation results for the recursive permutation list algorithm are obtained using an open source project available online \cite{DumerSource}.

In Fig. \ref{awgn}, we present the block error rate (BLER) of the considered decoders for a BI-AWGN channel.
For graph search decoding, we consider two scenarios. 
In the first scenario, plotted as solid blue curves, the number of iterations is chosen in such a way that the worst-case\footnote{We consider the worst-case running time of the graph search algorithm as in Theorem 1, i.e. the break statement does not terminate the main loop of Algorithm \ref{GreedyGraphSearchAlg} and each call of the \texttt{NextStepGreedy} function takes the worst-case running time. } running time of the proposed algorithm is similar to that of recursive permutation list decoding with the list of size $L$.
In the second scenario, plotted as dashed orange curves, we consider the number of iterations required to perform within 0.1 dB from the ML decoding lower bound that is constructed using an approach described in \cite{Dumer}.
Namely, we compare the metric of the output codeword of the graph search decoder with a large running time to that of the transmitted codeword.
If the metric of the decoder's output is better, then the ML decoder will also return an incorrect codeword.
Consequently, the ratio of such events gives a lower bound on the ML decoding performance.

\begin{figure*}[tb]

\centering
  \subfloat[$\mathcal{R}\left(3, 8\right)$\label{1a}]{%
       \begin{tikzpicture}[ scale = 0.68 ]
	\begin{semilogyaxis}[xlabel={$E_b/N_0 \left(\textrm{dB}\right)$}, ylabel=Block error rate, xmin=0.5, xmax=2.6, ymax=1e-1, ymin=1e-4, grid=both, yminorgrids=true,legend pos=south west, legend style={font=\footnotesize}, legend cell align=left, minor tick num = 4]	
\addplot[solid, thick, mark=x, mark options={solid}, red] coordinates {
(0.0000,3.747870e-01)
(0.2000,2.876028e-01)
(0.4000,2.109294e-01)
(0.6000,1.461459e-01)
(0.8000,9.536496e-02)
(1.0000,5.802583e-02)
(1.2000,3.353732e-02)
(1.4000,1.834207e-02)
(1.6000,8.831621e-03)
(1.8000,4.201972e-03)
(2.0000,1.624078e-03)
(2.2000,6.254636e-04)
(2.4000,2.339956e-04)
(2.6000,8.040431e-05)
};
\addlegendentry{RPL $L = 256$}
\addplot[solid, thick, mark=star, mark options={solid}, violet] coordinates {
(0.0000,5.053125e-01)
(0.2000,4.006452e-01)
(0.4000,3.116129e-01)
(0.6000,2.215625e-01)
(0.8000,1.578125e-01)
(1.0000,1.020968e-01)
(1.2000,6.067416e-02)
(1.4000,3.646667e-02)
(1.6000,1.797872e-02)
(1.8000,8.932384e-03)
(2.0000,3.546099e-03)
(2.2000,1.602564e-03)
(2.4000,6.215040e-04)
(2.6000,2.743142e-04)
(2.8000,1.046859e-04)
(3.0000,4.977601e-05)
};
\addlegendentry{RPA}
\addplot[solid, thick, mark=+, mark options={solid}, blue] coordinates {
(0.0000,3.115156e-01)
(0.2000,2.263125e-01)
(0.4000,1.583438e-01)
(0.6000,1.034687e-01)
(0.8000,6.395238e-02)
(1.0000,3.752381e-02)
(1.2000,2.008065e-02)
(1.4000,1.024779e-02)
(1.6000,4.635036e-03)
(1.8000,1.967742e-03)
(2.0000,8.525180e-04)
(2.2000,3.177843e-04)
(2.4000,1.150055e-04)
(2.6000,4.571984e-05)
};
\addlegendentry{GS $N = 40$}
\addplot[dashed, thick, mark=square, mark options={solid}, orange] coordinates {
(0.0000,2.807188e-01)
(0.2000,2.005625e-01)
(0.4000,1.360938e-01)
(0.6000,8.790625e-02)
(0.8000,5.087500e-02)
(1.0000,2.990323e-02)
(1.2000,1.586170e-02)
(1.4000,7.875776e-03)
(1.6000,3.191489e-03)
(1.8000,1.506494e-03)
(2.0000,5.943152e-04)
(2.2000,2.500000e-04)
(2.4000,9.695817e-05)
};
\addlegendentry{GS $N = 64$}
\addplot[dotted, thick, mark=o, mark options={solid}, black] coordinates {
(0.0000,2.374063e-01)
(0.2000,1.671250e-01)
(0.4000,1.114375e-01)
(0.6000,7.059375e-02)
(0.8000,4.087500e-02)
(1.0000,2.317742e-02)
(1.2000,1.277660e-02)
(1.4000,6.155280e-03)
(1.6000,2.351064e-03)
(1.8000,1.116883e-03)
(2.0000,4.832041e-04)
(2.2000,2.222222e-04)
(2.4000,9.125475e-05)
};
\addlegendentry{ML lower bound}

	\end{semilogyaxis}
	\end{tikzpicture}
       }
    \hfill
  \subfloat[$\mathcal{R}\left(4, 8\right)$\label{1b}]{%
        \begin{tikzpicture}[ scale = 0.68 ]
	\begin{semilogyaxis}[xlabel={$E_b/N_0 \left(\textrm{dB}\right)$}, ylabel=Block error rate, xmin=1.5, xmax=3.5, ymax=1e-1, ymin=1e-4, grid=both, yminorgrids=true,legend pos=south west, legend style={font=\footnotesize}, legend cell align=left, minor tick num = 4]	
\addplot[solid, thick, mark=x, mark options={solid}, red] coordinates {
(1.4000,2.726666e-01)
(1.6000,1.815503e-01)
(1.8000,1.122649e-01)
(2.0000,6.444717e-02)
(2.2000,3.440523e-02)
(2.4000,1.611701e-02)
(2.6000,7.037907e-03)
(2.8000,2.829216e-03)
(3.0000,1.114703e-03)
(3.2000,3.024668e-04)
(3.4000,9.780037e-05)
};
\addlegendentry{RPL $L = 256$}
\addplot[solid, thick, mark=star, mark options={solid}, violet] coordinates {
(1.6000,3.123529e-01)
(1.8000,2.252174e-01)
(2.0000,1.351351e-01)
(2.2000,8.079365e-02)
(2.4000,4.424779e-02)
(2.6000,2.216814e-02)
(2.8000,1.043659e-02)
(3.0000,4.227848e-03)
(3.2000,1.531394e-03)
(3.4000,5.851375e-04)
(3.6000,1.550000e-04)
(3.8000,4.500000e-05)
};
\addlegendentry{RPA}
\addplot[solid, thick, mark=+, mark options={solid}, blue] coordinates {
(1.4000,2.228750e-01)
(1.6000,1.405625e-01)
(1.8000,8.021875e-02)
(2.0000,4.359375e-02)
(2.2000,2.089247e-02)
(2.4000,9.549020e-03)
(2.6000,3.920635e-03)
(2.8000,1.525641e-03)
(3.0000,4.679487e-04)
(3.2000,1.615385e-04)
(3.4000,4.931893e-05)
};
\addlegendentry{GS $N = 56$}
\addplot[dashed, thick, mark=square, mark options={solid}, orange] coordinates {
(1.4000,1.756875e-01)
(1.6000,1.046875e-01)
(1.8000,5.740625e-02)
(2.0000,2.933871e-02)
(2.2000,1.378723e-02)
(2.4000,5.962162e-03)
(2.6000,2.265957e-03)
(2.8000,8.418972e-04)
(3.0000,3.252688e-04)
(3.2000,1.328588e-04)
(3.4000,4.267990e-05)
};
\addlegendentry{GS $N = 128$}
\addplot[dotted, thick, mark=o, mark options={solid}, black] coordinates {
(1.4000,1.410937e-01)
(1.6000,8.368750e-02)
(1.8000,4.484375e-02)
(2.0000,2.225806e-02)
(2.2000,1.073404e-02)
(2.4000,4.870270e-03)
(2.6000,1.968085e-03)
(2.8000,7.509881e-04)
(3.0000,2.983871e-04)
(3.2000,1.245552e-04)
(3.4000,4.218362e-05)
};
\addlegendentry{ML lower bound}

	\end{semilogyaxis}
	\end{tikzpicture}
        }
   \hfill
  \subfloat[$\mathcal{R}\left(5, 8\right)$\label{1c}]{%
        \begin{tikzpicture}[ scale = 0.68 ]
	\begin{semilogyaxis}[xlabel={$E_b/N_0 \left(\textrm{dB}\right)$}, ylabel=Block error rate, xmin=3.4, xmax=5.2, ymax=1e-1, ymin=1e-4, grid=both, yminorgrids=true,legend pos=south west, legend style={font=\footnotesize}, legend cell align=left, minor tick num = 4]	
\addplot[solid, thick, mark=x, mark options={solid}, red] coordinates {
(3.0000,2.493873e-01)
(3.2000,1.641556e-01)
(3.4000,9.991604e-02)
(3.6000,5.563311e-02)
(3.8000,2.885820e-02)
(4.0000,1.373152e-02)
(4.2000,5.934978e-03)
(4.4000,2.550350e-03)
(4.6000,9.688653e-04)
(4.8000,3.118075e-04)
(5.0000,1.243728e-04)
(5.2000,3.568912e-05)
};
\addlegendentry{RPL $L = 64$}

\addplot[solid, thick, mark=+, mark options={solid}, blue] coordinates {
(3.0000,3.395064e-01)
(3.2000,2.351452e-01)
(3.4000,1.521713e-01)
(3.6000,9.109787e-02)
(3.8000,4.996831e-02)
(4.0000,2.546264e-02)
(4.2000,1.216667e-02)
(4.4000,5.031447e-03)
(4.6000,1.950092e-03)
(4.8000,6.880416e-04)
(5.0000,2.305141e-04)
(5.2000,7.026477e-05)
};
\addlegendentry{GS $N = 12$}
\addplot[dashed, thick, mark=square, mark options={solid}, orange] coordinates {
(3.0000,1.898125e-01)
(3.2000,1.169687e-01)
(3.4000,6.390323e-02)
(3.6000,3.434375e-02)
(3.8000,1.673437e-02)
(4.0000,7.342246e-03)
(4.2000,3.030303e-03)
(4.4000,1.437722e-03)
(4.6000,5.663900e-04)
(4.8000,2.219178e-04)
(5.0000,7.610854e-05)
};
\addlegendentry{GS $N = 128$}
\addplot[dotted, thick, mark=o, mark options={solid}, black] coordinates {
(3.0000,1.758125e-01)
(3.2000,1.082188e-01)
(3.4000,5.925806e-02)
(3.6000,3.212500e-02)
(3.8000,1.545312e-02)
(4.0000,6.989305e-03)
(4.2000,2.954545e-03)
(4.4000,1.387900e-03)
(4.6000,5.622407e-04)
(4.8000,2.191781e-04)
(5.0000,7.544672e-05)
};
\addlegendentry{ML lower bound}

	\end{semilogyaxis}
	\end{tikzpicture}
        }
    \\
  \subfloat[$\mathcal{R}\left(2, 9\right)$\label{1d}]{%
       \begin{tikzpicture}[ scale = 0.68 ]
	\begin{semilogyaxis}[xlabel={$E_b/N_0 \left(\textrm{dB}\right)$}, ylabel=Block error rate, xmin=0, xmax=2.4, ymax=1e-1, ymin=1e-4, grid=both, yminorgrids=true,legend pos=south west, legend style={font=\footnotesize}, legend cell align=left, minor tick num = 4]	
\addplot[solid, thick, mark=x, mark options={solid}, red] coordinates {
(0.0000,1.159694e-01)
(0.2000,8.287842e-02)
(0.4000,5.669371e-02)
(0.6000,3.734713e-02)
(0.8000,2.395504e-02)
(1.0000,1.426853e-02)
(1.2000,7.809662e-03)
(1.4000,4.582430e-03)
(1.6000,2.230483e-03)
(1.8000,1.132111e-03)
(2.0000,5.471807e-04)
(2.2000,2.263108e-04)
(2.4000,9.606385e-05)
};
\addlegendentry{RPL $L = 256$}
\addplot[solid, thick, mark=star, mark options={solid}, violet] coordinates {
(0.0000,1.071875e-01)
(0.2000,7.162791e-02)
(0.4000,4.732877e-02)
(0.6000,3.240506e-02)
(0.8000,1.988506e-02)
(1.0000,1.177590e-02)
(1.2000,6.316940e-03)
(1.4000,3.190045e-03)
(1.6000,1.589404e-03)
(1.8000,8.130081e-04)
(2.0000,3.771712e-04)
(2.2000,1.581809e-04)
(2.4000,5.316578e-05)
};
\addlegendentry{RPA}
\addplot[solid, thick, mark=+, mark options={solid}, blue] coordinates {
(0.0000,9.014516e-02)
(0.2000,6.254839e-02)
(0.4000,4.245161e-02)
(0.6000,2.737097e-02)
(0.8000,1.646774e-02)
(1.0000,9.666667e-03)
(1.2000,5.645161e-03)
(1.4000,2.819149e-03)
(1.6000,1.506494e-03)
(1.8000,7.110390e-04)
(2.0000,3.447099e-04)
(2.2000,1.536189e-04)
(2.4000,5.537634e-05)
};
\addlegendentry{GS $N = 24$}
\addplot[dashed, thick, mark=square, mark options={solid}, orange] coordinates {
(0.0000,7.893750e-02)
(0.2000,5.334375e-02)
(0.4000,3.528125e-02)
(0.6000,2.231522e-02)
(0.8000,1.332258e-02)
(1.0000,7.564935e-03)
(1.2000,4.328125e-03)
(1.4000,2.169355e-03)
(1.6000,1.027523e-03)
(1.8000,4.790698e-04)
(2.0000,2.186147e-04)
(2.2000,9.674330e-05)
};
\addlegendentry{GS $N = 32$}
\addplot[dotted, thick, mark=o, mark options={solid}, black] coordinates {
(0.0000,6.431250e-02)
(0.2000,4.246875e-02)
(0.4000,2.756250e-02)
(0.6000,1.722826e-02)
(0.8000,1.002419e-02)
(1.0000,5.603896e-03)
(1.2000,3.062500e-03)
(1.4000,1.516129e-03)
(1.6000,7.110092e-04)
(1.8000,3.372093e-04)
(2.0000,1.471861e-04)
(2.2000,6.896552e-05)
};
\addlegendentry{ML lower bound}

	\end{semilogyaxis}
	\end{tikzpicture}
       }
    \hfill
  \subfloat[$\mathcal{R}\left(3, 9\right)$\label{1e}]{%
        \begin{tikzpicture}[ scale = 0.68 ]
	\begin{semilogyaxis}[xlabel={$E_b/N_0 \left(\textrm{dB}\right)$}, ylabel=Block error rate, xmin=-0.2, xmax=2, ymax=1e-1, ymin=1e-4, grid=both, yminorgrids=true,legend pos=south west, legend style={font=\footnotesize}, legend cell align=left, minor tick num = 4]	
\addplot[solid, thick, mark=x, mark options={solid}, red] coordinates {
(0.6000,1.745801e-01)
(0.8000,1.068340e-01)
(1.0000,6.147564e-02)
(1.2000,3.281510e-02)
(1.4000,1.530407e-02)
(1.6000,6.350327e-03)
(1.8000,2.399685e-03)
(2.0000,8.181294e-04)
(2.2000,2.634879e-04)
(2.4000,9.121832e-05)
};
\addlegendentry{RPL $L = 1024$}
\addplot[solid, thick, mark=star, mark options={solid}, violet] coordinates {
(0.0000,3.176190e-01)
(0.2000,2.096429e-01)
(0.4000,1.339474e-01)
(0.6000,8.063492e-02)
(0.8000,4.200000e-02)
(1.0000,1.838235e-02)
(1.2000,8.620690e-03)
(1.4000,3.215434e-03)
(1.6000,9.372071e-04)
(1.8000,3.150000e-04)
(2.0000,9.310987e-05)
};
\addlegendentry{RPA}
\addplot[solid, thick, mark=+, mark options={solid}, blue] coordinates {
(0.0000,2.470455e-01)
(0.2000,1.611765e-01)
(0.4000,9.989286e-02)
(0.6000,5.796667e-02)
(0.8000,2.980645e-02)
(1.0000,1.394565e-02)
(1.2000,6.047619e-03)
(1.4000,2.354839e-03)
(1.6000,7.302158e-04)
(1.8000,2.518703e-04)
(2.0000,5.238345e-05)
};
\addlegendentry{GS $N = 150$}
\addplot[dashed, thick, mark=square, mark options={solid}, orange] coordinates {
(0.0000,9.691667e-02)
(0.2000,5.103846e-02)
(0.4000,2.492683e-02)
(0.6000,1.164130e-02)
(0.8000,4.740741e-03)
(1.0000,1.685484e-03)
(1.2000,5.516304e-04)
(1.4000,1.658375e-04)
(1.6000,4.845155e-05)
};
\addlegendentry{GS $N = 1024$}
\addplot[dotted, thick, mark=o, mark options={solid}, black] coordinates {
(0.0000,6.491667e-02)
(0.2000,3.146154e-02)
(0.4000,1.560976e-02)
(0.6000,6.510870e-03)
(0.8000,2.564815e-03)
(1.0000,9.112903e-04)
(1.2000,3.097826e-04)
(1.4000,9.452736e-05)
(1.6000,3.296703e-05)
};
\addlegendentry{ML lower bound}

	\end{semilogyaxis}
	\end{tikzpicture}
        }
   \hfill
  \subfloat[$\mathcal{R}\left(4, 9\right)$\label{1f}]{%
        \begin{tikzpicture}[ scale = 0.68 ]
	\begin{semilogyaxis}[xlabel={$E_b/N_0 \left(\textrm{dB}\right)$}, ylabel=Block error rate, xmin=0, xmax=3, ymax=1e-1, ymin=1e-4, grid=both, yminorgrids=true,legend pos=south west, legend style={font=\footnotesize}, legend cell align=left, minor tick num = 4]	
\addplot[solid, thick, mark=x, mark options={solid}, red] coordinates {
(1.4000,2.468294e-01)
(1.6000,1.410980e-01)
(1.8000,7.415789e-02)
(2.0000,3.370539e-02)
(2.2000,1.314517e-02)
(2.4000,4.335949e-03)
(2.6000,1.231587e-03)
(2.8000,2.906894e-04)
(3.0000,5.786213e-05)
};
\addlegendentry{RPL $L = 1024$}
\addplot[solid, thick, mark=+, mark options={solid}, blue] coordinates {
(1.4000,1.743750e-01)
(1.6000,9.222581e-02)
(1.8000,4.280645e-02)
(2.0000,1.725806e-02)
(2.2000,5.586022e-03)
(2.4000,1.669355e-03)
(2.6000,3.311258e-04)
(2.8000,7.727621e-05)
};
\addlegendentry{GS $N = 220$}
\addplot[dashed, thick, mark=square, mark options={solid}, orange] coordinates {
(0.6000,1.746667e-01)
(0.8000,8.525000e-02)
(1.0000,3.537931e-02)
(1.2000,1.305195e-02)
(1.4000,4.166667e-03)
(1.6000,1.046632e-03)
(1.8000,2.442002e-04)
(2.0000,5.055612e-05)
};
\addlegendentry{GS $N = 2^{15}$}
\addplot[dotted, thick, mark=o, mark options={solid}, black] coordinates {
(0.6000,1.161667e-01)
(0.8000,5.308333e-02)
(1.0000,1.975862e-02)
(1.2000,6.714286e-03)
(1.4000,1.987500e-03)
(1.6000,5.544041e-04)
(1.8000,1.318681e-04)
(2.0000,3.842265e-05)
};
\addlegendentry{ML lower bound}

	\end{semilogyaxis}
	\end{tikzpicture}
        }
   \\ 
   \subfloat[$\mathcal{R}\left(5, 9\right)$\label{1g}]{%
       \begin{tikzpicture}[ scale = 0.68 ]
	\begin{semilogyaxis}[xlabel={$E_b/N_0 \left(\textrm{dB}\right)$}, ylabel=Block error rate, xmin=1.8, xmax=3.8, ymax=1e-1, ymin=1e-4, grid=both, yminorgrids=true,legend pos=south west, legend style={font=\footnotesize}, legend cell align=left, minor tick num = 4]	
\addplot[solid, thick, mark=x, mark options={solid}, red] coordinates {
(2.4000,2.366215e-01)
(2.6000,1.271557e-01)
(2.8000,5.911897e-02)
(3.0000,2.407077e-02)
(3.2000,8.442082e-03)
(3.4000,2.263629e-03)
(3.6000,5.216283e-04)
(3.8000,9.522957e-05)
};
\addlegendentry{RPL $L = 1024$}
\addplot[solid, thick, mark=+, mark options={solid}, blue] coordinates {
(2.4000,2.157143e-01)
(2.6000,1.038750e-01)
(2.8000,4.384615e-02)
(3.0000,1.540541e-02)
(3.2000,4.675926e-03)
(3.4000,1.294872e-03)
(3.6000,2.061224e-04)
(3.8000,2.909270e-05)
};
\addlegendentry{GS $N = 250$}
\addplot[dashed, thick, mark=square, mark options={solid}, orange] coordinates {
(2.0000,1.441429e-01)
(2.2000,6.052941e-02)
(2.4000,2.173913e-02)
(2.6000,6.655629e-03)
(2.8000,1.812500e-03)
(3.0000,4.376368e-04)
(3.2000,1.094092e-04)
(3.4000,2.950000e-05)
};
\addlegendentry{GS $N = 2^{14}$}
\addplot[dotted, thick, mark=o, mark options={solid}, black] coordinates {
(2.0000,1.014286e-01)
(2.2000,3.976471e-02)
(2.4000,1.486957e-02)
(2.6000,4.390728e-03)
(2.8000,1.223214e-03)
(3.0000,3.741794e-04)
(3.2000,9.846827e-05)
(3.4000,2.700000e-05)
};
\addlegendentry{ML lower bound}

	\end{semilogyaxis}
	\end{tikzpicture}
       }
    \hfill
  \subfloat[$\mathcal{R}\left(2, 10\right)$\label{1h}]{%
        \begin{tikzpicture}[ scale = 0.68 ]
	\begin{semilogyaxis}[xlabel={$E_b/N_0 \left(\textrm{dB}\right)$}, ylabel=Block error rate, xmin=0.4, xmax=2.4, ymax=1e-1, ymin=1e-4, grid=both, yminorgrids=true,legend pos=north east, legend style={font=\footnotesize}, legend cell align=left, minor tick num = 4]	
\addplot[solid, thick, mark=x, mark options={solid}, red] coordinates {
(0.0000,1.390073e-01)
(0.2000,1.044534e-01)
(0.4000,7.306434e-02)
(0.6000,4.810341e-02)
(0.8000,3.052910e-02)
(1.0000,1.816320e-02)
(1.2000,1.036418e-02)
(1.4000,5.292646e-03)
(1.6000,2.461894e-03)
(1.8000,1.196570e-03)
(2.0000,5.320942e-04)
(2.2000,2.086281e-04)
(2.4000,8.001775e-05)
};
\addlegendentry{RPL $L = 1024$}
\addplot[solid, thick, mark=star, mark options={solid}, violet] coordinates {
(0.0000,8.062500e-02)
(0.2000,5.087302e-02)
(0.4000,3.048485e-02)
(0.6000,1.702532e-02)
(0.8000,9.738806e-03)
(1.0000,4.377510e-03)
(1.2000,2.161017e-03)
(1.4000,1.150342e-03)
(1.6000,4.011887e-04)
(1.8000,1.495513e-04)
(2.0000,5.424063e-05)
};
\addlegendentry{RPA}
\addplot[solid, thick, mark=+, mark options={solid}, blue] coordinates {
(-0.6000,1.657667e-01)
(-0.4000,1.185333e-01)
(-0.2000,8.013333e-02)
(0.0000,5.306250e-02)
(0.2000,3.165625e-02)
(0.4000,1.912903e-02)
(0.6000,1.102174e-02)
(0.8000,5.777174e-03)
(1.0000,2.750000e-03)
(1.2000,1.157609e-03)
(1.4000,5.350000e-04)
(1.6000,2.239130e-04)
(1.8000,8.978873e-05)
};
\addlegendentry{GS $N = 90$}
\addplot[dashed, thick, mark=square, mark options={solid}, orange] coordinates {
(-0.6000,1.507333e-01)
(-0.4000,1.047333e-01)
(-0.2000,6.886667e-02)
(0.0000,4.453333e-02)
(0.2000,2.730645e-02)
(0.4000,1.597468e-02)
(0.6000,8.894309e-03)
(0.8000,3.907407e-03)
(1.0000,1.852459e-03)
(1.2000,8.559671e-04)
(1.4000,3.537415e-04)
(1.6000,1.595577e-04)
(1.8000,5.194805e-05)
};
\addlegendentry{GS $N = 128$}
\addplot[dotted, thick, mark=o, mark options={solid}, black] coordinates {
(-0.6000,1.289000e-01)
(-0.4000,8.703333e-02)
(-0.2000,5.593333e-02)
(0.0000,3.536667e-02)
(0.2000,2.096774e-02)
(0.4000,1.232911e-02)
(0.6000,6.569106e-03)
(0.8000,2.666667e-03)
(1.0000,1.270492e-03)
(1.2000,5.843621e-04)
(1.4000,2.244898e-04)
(1.6000,1.042654e-04)
(1.8000,3.584416e-05)
};
\addlegendentry{ML lower bound}

	\end{semilogyaxis}
	\end{tikzpicture}
        }
   \hfill
  \subfloat[$\mathcal{R}\left(2, 11\right)$\label{1i}]{%
        \begin{tikzpicture}[ scale = 0.68 ]
	\begin{semilogyaxis}[xlabel={$E_b/N_0 \left(\textrm{dB}\right)$}, ylabel=Block error rate, xmin=-1, xmax=2.5, ymax=1e-1, ymin=1e-4, grid=both, yminorgrids=true,legend pos=south west, legend style={font=\footnotesize}, legend cell align=left, minor tick num = 4]	
\addplot[solid, thick, mark=x, mark options={solid}, red] coordinates {
(1.0000,1.105084e-01)
(1.2000,7.310026e-02)
(1.4000,4.557607e-02)
(1.6000,2.750766e-02)
(1.8000,1.595002e-02)
(2.0000,8.309958e-03)
(2.2000,3.683112e-03)
(2.4000,1.645924e-03)
(2.6000,8.110427e-04)
(2.8000,3.128688e-04)
(3.0000,8.439452e-05)
};
\addlegendentry{RPL $L = 1024$}
\addplot[solid, thick, mark=star, mark options={solid}, violet] coordinates {
(0.0000,7.803030e-02)
(0.2000,4.608333e-02)
(0.4000,2.483092e-02)
(0.6000,1.312020e-02)
(0.8000,6.485788e-03)
(1.0000,2.774725e-03)
(1.2000,1.253012e-03)
(1.4000,4.247876e-04)
(1.6000,1.296758e-04)
(1.8000,2.995507e-05)
};
\addlegendentry{RPA}
\addplot[solid, thick, mark=+, mark options={solid}, blue] coordinates {
(-0.6000,3.673333e-01)
(-0.4000,2.964333e-01)
(-0.2000,2.350000e-01)
(0.0000,1.747667e-01)
(0.2000,1.283667e-01)
(0.4000,9.246667e-02)
(0.6000,6.173333e-02)
(0.8000,3.793333e-02)
(1.0000,2.504839e-02)
(1.2000,1.440659e-02)
(1.4000,7.675325e-03)
(1.6000,4.080645e-03)
(1.8000,1.844262e-03)
(2.0000,8.688525e-04)
(2.2000,2.958580e-04)
(2.4000,1.507463e-04)
(2.6000,4.678945e-05)
};
\addlegendentry{GS $N = 80$}
\addplot[dashed, thick, mark=square, mark options={solid}, orange] coordinates {
(-0.6000,1.009091e-01)
(-0.4000,6.746667e-02)
(-0.2000,4.029630e-02)
(0.0000,2.275000e-02)
(0.2000,1.234940e-02)
(0.4000,6.000000e-03)
(0.6000,2.955882e-03)
(0.8000,1.142045e-03)
(1.0000,4.391304e-04)
(1.2000,1.499250e-04)
(1.4000,4.800000e-05)
};
\addlegendentry{GS $N = 1024$}
\addplot[dotted, thick, mark=o, mark options={solid}, black] coordinates {
(-0.6000,8.972727e-02)
(-0.4000,5.940000e-02)
(-0.2000,3.400000e-02)
(0.0000,1.947917e-02)
(0.2000,1.021687e-02)
(0.4000,4.898204e-03)
(0.6000,2.411765e-03)
(0.8000,9.204545e-04)
(1.0000,3.065217e-04)
(1.2000,9.595202e-05)
(1.4000,2.950000e-05)
};
\addlegendentry{ML lower bound}

	\end{semilogyaxis}
	\end{tikzpicture}
        }
	
\caption{The block error rate performance of RM codes on a BI-AWGN channel. For the recursive permutation list algorithm, we use $L$ to denote the list size. For the graph search algorithm, we use $N$ to denote the number of iterations.} 
\label{awgn}
\end{figure*}
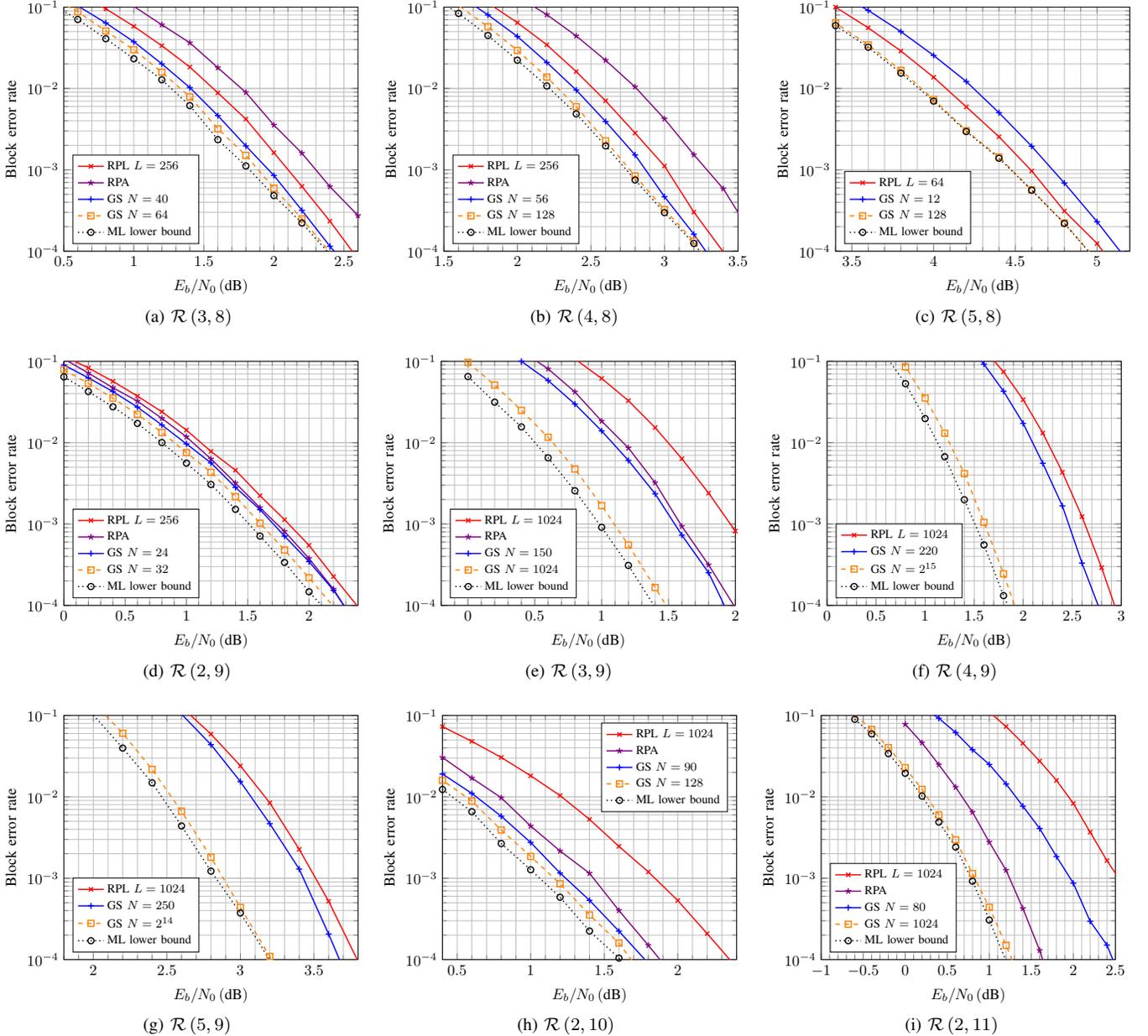

\begin{table*}[tb]
\caption{Comparison of average-case running time to decode one codeword between algorithms considered in Fig. \ref{awgn}. For the graph search and the RPA algorithms, the running time is estimated at $E_b/N_0$ required to reach a BLER of $10^{-4}$. The worst-case running time of the graph search algorithm is similar to the running time of the recursive permutation list algorithm.}
\begin{center}
\begin{tabular}{|c|c|c|c|c|c|c|c|c|c|}
\hline
Code & $\mathcal{R}\left(3,8\right)$ & $\mathcal{R}\left(4,8\right)$ & $\mathcal{R}\left(5,8\right)$ & $\mathcal{R}\left(2,9\right)$ & $\mathcal{R}\left(3,9\right)$ & $\mathcal{R}\left(4,9\right)$ & $\mathcal{R}\left(5,9\right)$ & $\mathcal{R}\left(2,10\right)$ & $\mathcal{R}\left(2,11\right)$ \\
\hline
RPL & 3.9ms & 5.5ms & 1.5ms & 4ms & 29ms & 46ms & 53ms & 35ms & 66ms\\
\hline
RPA & 196ms & 9.1s & -- & 14ms & 2.1s & -- & -- & 58ms & 227ms \\
\hline
GS & 2.4ms & 1ms & 0.2ms & 3.1ms & 10ms & 4.8ms & 4ms & 24ms & 45ms\\
\hline
\end{tabular}
\label{runningTime}
\end{center}
\end{table*}

We can see that, in most cases, the graph search decoder outperforms the recursive permutation list decoder with a similar running time by approximately 0.1 dB.
The gap increases to 0.4 dB and 0.6 dB for $\mathcal{R}\left(3, 9\right)$ and the second-order RM codes of length larger than 512, respectively.
However, the recursive permutation list algorithm performs better than the graph search algorithm for a high-rate $\mathcal{R}\left(5, 8\right)$.
Observe that for $\mathcal{R}\left(5, 8\right)$ we consider the recursive permutation list algorithm with a relatively small list of size 64, while for the rest cases the list size is at least 256.
Therefore, from Fig. \ref{awgn}, we notice that the graph search decoder is competitive with the recursive permutation list decoder with a list of size greater than 128.

RPA decoding outperforms graph search decoding considered in the first scenario only for $\mathcal{R}\left(2, 11\right)$.
However, it is done at the cost of a much longer running time.
In Table \ref{runningTime}, we present the running time to decode one codeword for the recursive permutation list, the RPA, and the graph search algorithms considered in Fig. \ref{awgn}.
Note that, for graph search decoding, results are presented for the first scenario.
Since the running time of the graph search and the RPA algorithms depends on the noise level, we only present the running time at a BLER of $10^{-4}$ for these two cases.
Recall that the worst-case running time of graph search decoding coincides with the running time of the recursive permutation list algorithm.
From Table \ref{runningTime}, we can see that the running time of RPA decoding is much longer compared to the proposed algorithm.
For $\mathcal{R}\left(2, 11\right)$, the graph search decoder with 1024 iterations outperforms RPA decoding by approximately 0.4 dB at a BLER of $10^{-4}$, while the average-case running time of these two algorithms is the same. 
Thus, from Fig. \ref{awgn}, we can conclude that the proposed algorithm outperforms the RPA decoder with a similar average-case running time.

In \cite{MinYe}, a simplified version of RPA decoding is proposed for high-rate RM codes.
For instance, a list version of simplified RPA decoding is demonstrated to achieve the ML decoding performance for $\mathcal{R}\left(5, 8\right)$. 
It takes approximately 13 ms at a BLER of $10^{-4}$.
The graph search decoder with 128 iterations also achieves the ML decoding performance, but it takes only 3 ms on average at a BLER of $10^{-4}$.
Another simplified version of RPA decoding, called sparse multi-decoder RPA (SRPA), is proposed in \cite{SRPA}.
The SRPA algorithm allows decreasing the running time of RPA decoding up to four times for $\mathcal{R}\left(2, 9\right)$ at the cost of 0.15 dB loss and up to eight times for $\mathcal{R}\left(3, 8\right)$ without any performance loss.
Thus, the running time of the SRPA algorithm and the graph search algorithm for $\mathcal{R}\left(2, 9\right)$ is similar.
However, due to 0.15 dB performance loss, the graph search decoding demonstrates slightly better performance.
In the case of $\mathcal{R}\left(3, 8\right)$, the running time of SRPA decoding is much longer compared to the proposed algorithm.
Hence, simplified versions of RPA decoding are also not competitive with the graph search algorithm.

From Fig. \ref{awgn}, we can see that the graph search decoder requires a large number of iterations to perform 0.1 dB from ML for RM codes of length 512 and order greater than 3.
For instance, the required number of iterations equals $2^{15}$ for $\mathcal{R}\left(4, 9\right)$ resulting in a very long running time in the worst case.
However, the average-case running time of the proposed algorithm is much smaller.
In Figs. \ref{runningTime39} and \ref{runningTime49} we compare the average-case running time as a function of $E_b/N_0$ for $\mathcal{R}\left(3, 9\right)$ and $\mathcal{R}\left(4, 9\right)$, respectively.
Note that the graph search decoder approaching ML performance takes less time to decode than the recursive permutation list algorithm in high signal-to-noise region.
The average-case running time of the graph search decoder can be further reduced using a CRC.
Recall that the proposed algorithm is terminated if a codeword with the best metric found so far satisfies the CRC.
In our simulations, we consider the 24-bits CRC with the generator polynomial 
\begin{equation*}
x^{24} + x^{23} + x^{21} + x^{20} + x^{17} + x^{15} + x^{13} + x^{12}+ x^8 + x^4 + x^2 + x + 1.
\end{equation*}
Such a long CRC guarantees that undetected errors, i.e. the codeword that satisfies the CRC is not a valid one, do not affect the performance for a BLER above $10^{-4}$.
Hence, the results presented in Figs. \ref{runningTime39} and \ref{runningTime49} can be regarded as the best possible improvement of the average-case running time if we have a genie that stops decoding immediately after a valid codeword is found. 
Although this technique allows reducing the average-case running time, the 24-bits CRC  results in a huge rate loss for the considered cases.
Thus, a further improvement of the proposed algorithm is required to decrease the average-case running time using a reasonable length CRC with negligible performance loss.

\begin{figure}[t]
\centering
\begin{tikzpicture}[ scale = 0.9 ]
	\begin{semilogyaxis}[xlabel={$E_b/N_0 \left(\textrm{dB}\right)$}, ylabel=Running time (s), xmin=0, xmax=2, ymax=6, ymin=6e-4, grid=both, yminorgrids=true,legend pos=north west, legend style={font=\footnotesize}, legend cell align=left, minor tick num = 4]	
\addplot[solid, thick, mark=x, mark options={solid}, red] coordinates {
(0.600,3.610918e-02)
(0.800,3.603699e-02)
(1.000,3.628957e-02)
(1.200,3.648063e-02)
(1.400,3.657724e-02)
(1.600,3.667639e-02)
(1.800,3.672617e-02)
(2.000,3.715150e-02)
(2.200,3.732920e-02)
(2.400,3.757632e-02)
};
\addlegendentry{RPL $L = 1024$}
\addplot[solid, thick, mark=star, mark options={solid}, violet] coordinates {
(0.000,5.163051e+00)
(0.200,4.731213e+00)
(0.400,4.400298e+00)
(0.600,4.003377e+00)
(0.800,3.226581e+00)
(1.000,2.903756e+00)
(1.200,2.703663e+00)
(1.400,2.584979e+00)
(1.600,2.657296e+00)
(1.800,2.452090e+00)
(2.000,2.396307e+00)
};
\addlegendentry{RPA}
\addplot[solid, thick, mark=+, mark options={solid}, blue] coordinates {
(0.000,2.160047e-02)
(0.200,2.131989e-02)
(0.400,2.171089e-02)
(0.600,2.133375e-02)
(0.800,2.037369e-02)
(1.000,1.908948e-02)
(1.200,1.747599e-02)
(1.400,1.558244e-02)
(1.600,1.360934e-02)
(1.800,1.147791e-02)
(2.000,9.660620e-03)
};
\addlegendentry{GS $N = 150$}
\addplot[solid, thick, mark=square, mark options={solid}, orange] coordinates {
(0.000,8.822037e-02)
(0.200,7.671514e-02)
(0.400,6.211095e-02)
(0.600,5.086709e-02)
(0.800,3.990388e-02)
(1.000,3.092275e-02)
(1.200,2.411974e-02)
(1.400,1.876108e-02)
(1.600,1.479175e-02)
};
\addlegendentry{GS $N = 1024$}
\addplot[dashed, thick, mark=+, mark options={solid}, blue] coordinates {
(0.000,8.239431e-03)
(0.200,6.517048e-03)
(0.400,5.070499e-03)
(0.600,3.831555e-03)
(0.800,2.840967e-03)
(1.000,2.058115e-03)
(1.200,1.525988e-03)
(1.400,1.117978e-03)
(1.600,8.349128e-04)
(1.800,6.240183e-04)
(2.000,4.731738e-04)
};
\addlegendentry{GS $N = 150$ (CRC)}
\addplot[dashed, thick, mark=square, mark options={solid}, orange] coordinates {
(0.000,2.298654e-02)
(0.200,1.506034e-02)
(0.400,9.485985e-03)
(0.600,6.015889e-03)
(0.800,3.874501e-03)
(1.000,2.467009e-03)
(1.200,1.684546e-03)
(1.400,1.163455e-03)
(1.600,8.446966e-04)
};
\addlegendentry{GS $N = 1024$ (CRC)}

	\end{semilogyaxis}
	\end{tikzpicture}
\caption{The average-case running time of $\mathcal{R}\left(3, 9\right)$ decoding on a BI-AWGN channel.}
\label{runningTime39}
\end{figure}
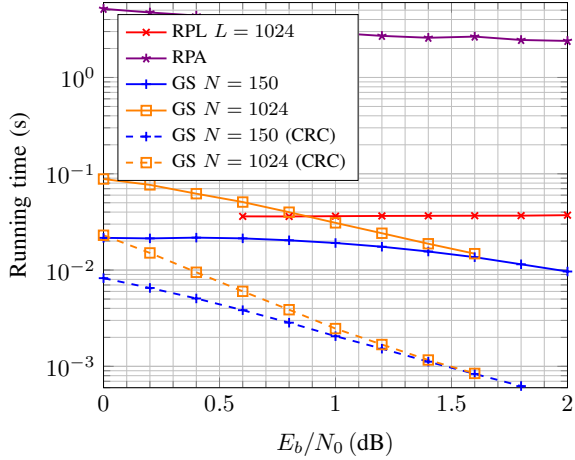

\begin{figure}[t]
\centering
\begin{tikzpicture}[ scale = 0.9 ]
	\begin{semilogyaxis}[xlabel={$E_b/N_0 \left(\textrm{dB}\right)$}, ylabel=Running time (s), xmin=0.6, xmax=3, ymax=3, ymin=6e-4, grid=both, yminorgrids=true,legend pos=north east, legend style={font=\footnotesize}, legend cell align=left, minor tick num = 4]	
\addplot[solid, thick, mark=x, mark options={solid}, red] coordinates {
(1.400,4.983069e-02)
(1.600,5.067006e-02)
(1.800,5.098758e-02)
(2.000,5.166807e-02)
(2.200,5.229967e-02)
(2.400,5.319768e-02)
(2.600,5.384504e-02)
(2.800,4.883615e-02)
(3.000,4.959127e-02)
};
\addlegendentry{RPL $L = 1024$}
\addplot[solid, thick, mark=+, mark options={solid}, blue] coordinates {
(1.400,2.638013e-02)
(1.600,2.251280e-02)
(1.800,1.836596e-02)
(2.000,1.448362e-02)
(2.200,1.094163e-02)
(2.400,8.183105e-03)
(2.600,6.180101e-03)
(2.800,4.617713e-03)
};
\addlegendentry{GS $N = 220$}
\addplot[solid, thick, mark=square, mark options={solid}, orange] coordinates {
(0.600,2.133472e+00)
(0.800,1.446850e+00)
(1.000,8.695518e-01)
(1.200,4.397261e-01)
(1.400,2.220544e-01)
(1.600,9.365897e-02)
(1.800,4.016416e-02)
(2.000,2.087183e-02)
};
\addlegendentry{GS $N = 2^{15}$}
\addplot[dashed, thick, mark=+, mark options={solid}, blue] coordinates {
(1.400,1.194665e-02)
(1.600,8.519181e-03)
(1.800,5.819517e-03)
(2.000,3.899534e-03)
(2.200,2.562395e-03)
(2.400,1.693636e-03)
(2.600,1.066481e-03)
(2.800,6.852950e-04)
};
\addlegendentry{GS $N = 220$ (CRC)}
\addplot[dashed, thick, mark=square, mark options={solid}, orange] coordinates {
(0.600,9.940621e-01)
(0.800,5.730542e-01)
(1.000,2.903275e-01)
(1.200,1.300597e-01)
(1.400,5.363423e-02)
(1.600,2.127696e-02)
(1.800,9.705992e-03)
(2.000,4.875927e-03)
};
\addlegendentry{GS $N = 2^{15}$ (CRC)}

	\end{semilogyaxis}
	\end{tikzpicture}
\caption{The average-case running time of $\mathcal{R}\left(4, 9\right)$ decoding on a BI-AWGN channel.}
\label{runningTime49}
\end{figure}
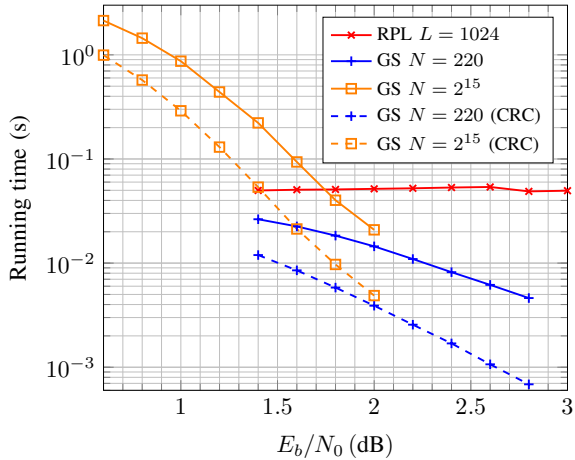

In Fig. \ref{bsc}, we plot the block error probability for a binary symmetric channel (BSC).
Note that we consider a simplified version of RPA decoding proposed for a BSC (see Algorithm 1 in \cite{MinYe}), while for graph search decoding and recursive permutation list decoding we use hard-input versions of the algorithms considered for a BI-AWGN channel.
As in Fig. \ref{awgn}, we use solid blue curves to report results of the graph search algorithm with the worst-case running time that is similar to that of the recursive permutation list algorithm and we use dashed orange curves to report results of the graph search algorithm approaching the ML performance.
In contrast to a BI-AWGN channel, the graph search decoder requires a moderate number of iterations ($N \leq 1024$) to approach the performance of ML decoding for a BSC.
For RM codes of length greater than 256, the proposed algorithm outperforms the recursive permutation list algorithm with the same worst-case running time, while these decoders perform similarly for RM codes of length 256.
Furthermore, the graph search algorithm outperforms RPA decoding and, at the same time, takes a shorter average-case running time.
For instance, the average-case running time of RPA decoding for $\mathcal{R}\left(3, 8\right)$ at a BLER of $10^{-4}$ is 105 ms, while the proposed algorithm with 16 iterations takes only 0.9 ms.

\begin{figure*}[tb]

\centering
  \subfloat[$\mathcal{R}\left(3, 8\right)$\label{2a}]{%
       \begin{tikzpicture}[ scale = 0.68 ]
	\begin{semilogyaxis}[xlabel=Crossover probability, ylabel=Block error rate, xmin=0.07, xmax=0.13, ymax=1e-1, ymin=1e-4, grid=both, yminorgrids=true,legend pos=south east, legend style={font=\footnotesize}, legend cell align=left, xticklabels={0.07,0.08,0.09,0.10,0.11,0.12}, xtick={0.07, 0.08,0.09,0.10,0.11,0.12}]	
\addplot[solid, thick, mark=x, mark options={solid}, red] coordinates {
(0.1500,4.211362e-01)
(0.1450,3.438278e-01)
(0.1400,2.708658e-01)
(0.1350,2.074050e-01)
(0.1300,1.525472e-01)
(0.1250,1.067594e-01)
(0.1200,7.138840e-02)
(0.1150,4.537379e-02)
(0.1100,2.736804e-02)
(0.1050,1.551301e-02)
(0.1000,8.432729e-03)
(0.0950,4.087032e-03)
(0.0900,1.808297e-03)
(0.0850,8.680162e-04)
(0.0800,3.386552e-04)
(0.0750,1.219809e-04)
(0.0700,3.921988e-05)
};
\addlegendentry{RPL $L = 64$}
\addplot[solid, thick, mark=star, mark options={solid}, violet] coordinates {
(0.1500,5.400000e-01)
(0.1450,4.455556e-01)
(0.1400,3.632143e-01)
(0.1350,2.837500e-01)
(0.1300,2.153448e-01)
(0.1250,1.537879e-01)
(0.1200,1.045833e-01)
(0.1150,6.503226e-02)
(0.1100,3.790614e-02)
(0.1050,2.115304e-02)
(0.1000,1.079827e-02)
(0.0950,5.106383e-03)
(0.0900,2.263514e-03)
(0.0850,9.082652e-04)
(0.0800,3.376097e-04)
(0.0750,1.367054e-04)
(0.0700,4.497302e-05)
};
\addlegendentry{RPA}
\addplot[solid, thick, mark=+, mark options={solid}, blue] coordinates {
(0.1500,4.497550e-01)
(0.1450,3.703372e-01)
(0.1400,2.966921e-01)
(0.1350,2.290297e-01)
(0.1300,1.698107e-01)
(0.1250,1.208316e-01)
(0.1200,8.149128e-02)
(0.1150,5.222222e-02)
(0.1100,3.213782e-02)
(0.1050,1.845513e-02)
(0.1000,1.006028e-02)
(0.0950,5.117438e-03)
(0.0900,2.416667e-03)
(0.0850,1.078571e-03)
(0.0800,4.898649e-04)
(0.0750,1.837748e-04)
(0.0700,4.997526e-05)
};
\addlegendentry{GS $N = 7$}
\addplot[dashed, thick, mark=square, mark options={solid}, orange] coordinates {
(0.1400,2.122937e-01)
(0.1350,1.545645e-01)
(0.1300,1.085000e-01)
(0.1250,7.270690e-02)
(0.1200,4.696875e-02)
(0.1150,2.898958e-02)
(0.1100,1.672917e-02)
(0.1050,8.928000e-03)
(0.1000,4.552083e-03)
(0.0950,2.229167e-03)
(0.0900,9.444444e-04)
(0.0850,4.693141e-04)
(0.0800,2.092257e-04)
(0.0750,7.286096e-05)
};
\addlegendentry{GS $N = 16$}
\addplot[dotted, thick, mark=o, mark options={solid}, black] coordinates {
(0.1400,1.786825e-01)
(0.1350,1.287661e-01)
(0.1300,8.957258e-02)
(0.1250,5.962069e-02)
(0.1200,3.807292e-02)
(0.1150,2.355208e-02)
(0.1100,1.366667e-02)
(0.1050,7.524000e-03)
(0.1000,3.885417e-03)
(0.0950,1.968750e-03)
(0.0900,8.452381e-04)
(0.0850,4.151625e-04)
(0.0800,1.960461e-04)
(0.0750,7.085561e-05)
};
\addlegendentry{ML lower bound}

	\end{semilogyaxis}
	\end{tikzpicture}
       }
    \hfill
  \subfloat[$\mathcal{R}\left(4, 8\right)$\label{2b}]{%
        \begin{tikzpicture}[ scale = 0.68 ]
	\begin{semilogyaxis}[xlabel=Crossover probability, ylabel=Block error rate,, ylabel=Block error rate, xmin=0.015, xmax=0.04, ymax=1e-1, ymin=1e-4, grid=both, yminorgrids=true,legend pos=south east, legend style={font=\footnotesize}, legend cell align=left,xticklabels={0.02,0.03,0.04}, xtick={0.02,0.03,0.04}, scaled x ticks=false]
\addplot[solid, thick, mark=x, mark options={solid}, red] coordinates {
(0.0800,8.262505e-01)
(0.0750,7.503356e-01)
(0.0700,6.564215e-01)
(0.0650,5.479300e-01)
(0.0600,4.305801e-01)
(0.0550,3.153004e-01)
(0.0500,2.107870e-01)
(0.0450,1.266309e-01)
(0.0400,6.645244e-02)
(0.0350,2.933199e-02)
(0.0300,1.074044e-02)
(0.0250,2.905259e-03)
(0.0200,4.991459e-04)
(0.0150,5.718334e-05)
};
\addlegendentry{RPL $L = 32$}
\addplot[solid, thick, mark=star, mark options={solid}, violet] coordinates {
(0.0600,6.300000e-01)
(0.0550,4.820833e-01)
(0.0500,3.413333e-01)
(0.0450,2.110417e-01)
(0.0400,1.120879e-01)
(0.0350,5.102041e-02)
(0.0300,1.865922e-02)
(0.0250,4.625347e-03)
(0.0200,8.116883e-04)
(0.0150,9.455200e-05)
};
\addlegendentry{RPA}
\addplot[solid, thick, mark=+, mark options={solid}, blue] coordinates {
(0.0600,4.151410e-01)
(0.0550,2.978539e-01)
(0.0500,1.945246e-01)
(0.0450,1.137793e-01)
(0.0400,5.746543e-02)
(0.0350,2.470927e-02)
(0.0300,8.254545e-03)
(0.0250,2.130199e-03)
(0.0200,4.077025e-04)
(0.0150,4.412448e-05)
};
\addlegendentry{GS $N = 5$}
\addplot[dashed, thick, mark=square, mark options={solid}, orange] coordinates {
(0.0500,1.552578e-01)
(0.0450,8.907595e-02)
(0.0400,4.428191e-02)
(0.0350,1.919457e-02)
(0.0300,6.665493e-03)
(0.0250,1.777506e-03)
(0.0200,3.643927e-04)
(0.0150,4.008667e-05)
};
\addlegendentry{GS $N = 16$}
\addplot[dotted, thick, mark=o, mark options={solid}, black] coordinates {
(0.0500,1.504141e-01)
(0.0450,8.671519e-02)
(0.0400,4.326064e-02)
(0.0350,1.891403e-02)
(0.0300,6.602113e-03)
(0.0250,1.762836e-03)
(0.0200,3.643927e-04)
(0.0150,4.008667e-05)
};
\addlegendentry{ML lower bound}

	\end{semilogyaxis}
	\end{tikzpicture}
        }
   \hfill
  \subfloat[$\mathcal{R}\left(3, 9\right)$\label{2c}]{%
        \begin{tikzpicture}[ scale = 0.68 ]
	\begin{semilogyaxis}[xlabel=Crossover probability, ylabel=Block error rate,, ylabel=Block error rate, xmin=0.13, xmax=0.2, ymax=1e-1, ymin=1e-4, grid=both, yminorgrids=true,legend pos=south east, legend style={font=\footnotesize}, legend cell align=left, , xticklabels={0.13,0.14,0.15,0.16,0.17,0.18,0.19}, xtick={0.13,0.14,0.15,0.16,0.17,0.18,0.19}]
\addplot[solid, thick, mark=x, mark options={solid}, red] coordinates {
(0.2500,9.918170e-01)
(0.2450,9.847515e-01)
(0.2400,9.714451e-01)
(0.2350,9.535287e-01)
(0.2300,9.239332e-01)
(0.2250,8.860057e-01)
(0.2200,8.328110e-01)
(0.2150,7.679379e-01)
(0.2100,6.886792e-01)
(0.2050,6.009810e-01)
(0.2000,5.093615e-01)
(0.1950,4.074614e-01)
(0.1900,3.160530e-01)
(0.1850,2.324609e-01)
(0.1800,1.632902e-01)
(0.1750,1.072890e-01)
(0.1700,6.538948e-02)
(0.1650,3.678182e-02)
(0.1600,1.937897e-02)
(0.1550,9.387172e-03)
(0.1500,4.740222e-03)
(0.1450,1.923005e-03)
(0.1400,6.597399e-04)
(0.1350,2.484836e-04)
(0.1300,7.379939e-05)
};
\addlegendentry{RPL $L = 1024$}
\addplot[solid, thick, mark=star, mark options={solid}, violet] coordinates {
(0.2000,9.028571e-01)
(0.1950,8.371429e-01)
(0.1900,7.678571e-01)
(0.1850,6.568750e-01)
(0.1800,5.411538e-01)
(0.1750,4.196875e-01)
(0.1700,3.051351e-01)
(0.1650,2.054000e-01)
(0.1600,1.280247e-01)
(0.1550,7.304965e-02)
(0.1500,3.944882e-02)
(0.1450,1.953125e-02)
(0.1400,8.298755e-03)
(0.1350,3.105590e-03)
(0.1300,1.062699e-03)
(0.1250,3.539197e-04)
(0.1200,8.796481e-05)
};
\addlegendentry{RPA}
\addplot[solid, thick, mark=+, mark options={solid}, blue] coordinates {
(0.1900,1.294677e-01)
(0.1850,7.948387e-02)
(0.1800,4.619355e-02)
(0.1750,2.431667e-02)
(0.1700,1.269149e-02)
(0.1650,6.135135e-03)
(0.1600,2.659574e-03)
(0.1550,1.134409e-03)
(0.1500,3.559322e-04)
(0.1450,1.645161e-04)
(0.1400,3.716551e-05)
};
\addlegendentry{GS $N = 150$}
\addplot[dashed, thick, mark=square, mark options={solid}, orange] coordinates {
(0.2000,2.213333e-01)
(0.1950,1.517500e-01)
(0.1900,9.664706e-02)
(0.1850,5.626667e-02)
(0.1800,3.063333e-02)
(0.1750,1.645161e-02)
(0.1700,8.266129e-03)
(0.1650,4.000000e-03)
(0.1600,1.549020e-03)
(0.1550,6.590164e-04)
(0.1500,2.321429e-04)
(0.1450,1.069519e-04)
(0.1400,2.790234e-05)
};
\addlegendentry{GS $N = 512$}
\addplot[dotted, thick, mark=o, mark options={solid}, black] coordinates {
(0.2000,2.045000e-01)
(0.1950,1.385000e-01)
(0.1900,8.676471e-02)
(0.1850,4.953333e-02)
(0.1800,2.691667e-02)
(0.1750,1.464516e-02)
(0.1700,7.209677e-03)
(0.1650,3.548387e-03)
(0.1600,1.398693e-03)
(0.1550,6.032787e-04)
(0.1500,1.919643e-04)
(0.1450,9.732620e-05)
(0.1400,2.640757e-05)
};
\addlegendentry{ML lower bound}

	\end{semilogyaxis}
	\end{tikzpicture}
        }
    \\
  \subfloat[$\mathcal{R}\left(4, 9\right)$\label{2d}]{%
       \begin{tikzpicture}[ scale = 0.68 ]
	\begin{semilogyaxis}[xlabel=Crossover probability, ylabel=Block error rate,, ylabel=Block error rate, xmin=0.045, xmax=0.1, ymax=1e-1, ymin=1e-4, grid=both, yminorgrids=true,legend pos=south east, legend style={font=\footnotesize}, legend cell align=left,xticklabels={0.05,0.06,0.07,0.08, 0.09}, xtick={0.05,0.06,0.07,0.08, 0.09}, scaled x ticks=false]
\addplot[solid, thick, mark=x, mark options={solid}, red] coordinates {
(0.1000,5.818286e-01)
(0.0950,4.485820e-01)
(0.0900,3.222133e-01)
(0.0850,2.157395e-01)
(0.0800,1.208426e-01)
(0.0750,5.843747e-02)
(0.0700,2.678611e-02)
(0.0650,9.189479e-03)
(0.0600,2.613450e-03)
(0.0550,6.700990e-04)
(0.0500,1.205538e-04)
(0.0450,1.894615e-05)
};
\addlegendentry{RPL $L = 1024$}
\addplot[solid, thick, mark=+, mark options={solid}, blue] coordinates {
(0.1300,9.230000e-01)
(0.1250,8.642500e-01)
(0.1200,7.870000e-01)
(0.1150,6.850000e-01)
(0.1100,5.495000e-01)
(0.1050,4.107500e-01)
(0.1000,2.795000e-01)
(0.0950,1.662500e-01)
(0.0900,9.278125e-02)
(0.0850,4.564062e-02)
(0.0800,1.969863e-02)
(0.0750,7.225989e-03)
(0.0700,2.248120e-03)
(0.0650,7.011494e-04)
(0.0600,2.054545e-04)
(0.0550,4.759615e-05)
};
\addlegendentry{GS $N = 220$}
\addplot[dashed, thick, mark=square, mark options={solid}, orange] coordinates {
(0.1000,2.520000e-01)
(0.0950,1.451000e-01)
(0.0900,7.775000e-02)
(0.0850,3.887500e-02)
(0.0800,1.618462e-02)
(0.0750,6.034286e-03)
(0.0700,1.975610e-03)
(0.0650,6.410256e-04)
(0.0600,1.959459e-04)
(0.0550,4.586241e-05)
};
\addlegendentry{GS $N = 512$}
\addplot[dotted, thick, mark=o, mark options={solid}, black] coordinates {
(0.1000,2.250000e-01)
(0.0950,1.283500e-01)
(0.0900,6.882143e-02)
(0.0850,3.379687e-02)
(0.0800,1.453846e-02)
(0.0750,5.388571e-03)
(0.0700,1.821138e-03)
(0.0650,6.217949e-04)
(0.0600,1.942568e-04)
(0.0550,4.436690e-05)
};
\addlegendentry{ML lower bound}

	\end{semilogyaxis}
	\end{tikzpicture}
       }
    \hfill
  \subfloat[$\mathcal{R}\left(2, 10\right)$\label{2e}]{%
        \begin{tikzpicture}[ scale = 0.68 ]
	\begin{semilogyaxis}[xlabel=Crossover probability, ylabel=Block error rate,, ylabel=Block error rate, xmin=0.29, xmax=0.34, ymax=1e-1, ymin=1e-4, grid=both, yminorgrids=true,legend pos=north west, legend style={font=\footnotesize}, legend cell align=left,xticklabels={0.29, 0.3, 0.31, 0.32, 0.33, 0.34}, xtick={0.29, 0.3, 0.31, 0.32, 0.33, 0.34}, scaled x ticks=false]

\addplot[solid, thick, mark=x, mark options={solid}, red] coordinates {
(0.3500,2.991347e-01)
(0.3450,2.068175e-01)
(0.3400,1.384534e-01)
(0.3350,8.616268e-02)
(0.3300,4.818910e-02)
(0.3250,2.484749e-02)
(0.3200,1.247963e-02)
(0.3150,5.347622e-03)
(0.3100,2.388684e-03)
(0.3050,9.290202e-04)
(0.3000,3.683431e-04)
(0.2950,1.209146e-04)
(0.2900,2.948564e-05)
};
\addlegendentry{RPL $L = 1024$}
\addplot[solid, thick, mark=star, mark options={solid}, violet] coordinates {
(0.3600,7.542667e-01)
(0.3550,6.439394e-01)
(0.3500,5.265152e-01)
(0.3450,3.934848e-01)
(0.3400,2.715152e-01)
(0.3350,1.720588e-01)
(0.3300,9.900000e-02)
(0.3250,5.210046e-02)
(0.3200,2.500000e-02)
(0.3150,1.106987e-02)
(0.3100,4.404396e-03)
(0.3050,1.520772e-03)
(0.3000,4.807462e-04)
(0.2950,1.342459e-04)
(0.2900,3.288490e-05)
};
\addlegendentry{RPA}
\addplot[solid, thick, mark=+, mark options={solid}, blue] coordinates {
(0.3500,1.604167e-01)
(0.3450,9.580645e-02)
(0.3400,5.291935e-02)
(0.3350,2.580645e-02)
(0.3300,1.255319e-02)
(0.3250,5.467742e-03)
(0.3200,2.082474e-03)
(0.3150,7.142857e-04)
(0.3100,2.304147e-04)
(0.3050,7.535795e-05)
};
\addlegendentry{GS $N = 90$}
\addplot[dashed, thick, mark=square, mark options={solid}, orange] coordinates {
(0.3600,3.380000e-01)
(0.3550,2.205000e-01)
(0.3500,1.376316e-01)
(0.3450,7.705000e-02)
(0.3400,4.203226e-02)
(0.3350,1.998276e-02)
(0.3300,9.198276e-03)
(0.3250,3.383333e-03)
(0.3200,1.324675e-03)
(0.3150,4.486607e-04)
(0.3100,1.531394e-04)
(0.3050,4.293560e-05)
};
\addlegendentry{GS $N = 256$}
\addplot[dotted, thick, mark=o, mark options={solid}, black] coordinates {
(0.3600,3.325000e-01)
(0.3550,2.168333e-01)
(0.3500,1.347368e-01)
(0.3450,7.540000e-02)
(0.3400,4.077419e-02)
(0.3350,1.950000e-02)
(0.3300,8.793103e-03)
(0.3250,3.216667e-03)
(0.3200,1.272727e-03)
(0.3150,4.196429e-04)
(0.3100,1.439510e-04)
(0.3050,4.143784e-05)
};
\addlegendentry{ML lower bound}

	\end{semilogyaxis}
	\end{tikzpicture}
        }
   \hfill
  \subfloat[$\mathcal{R}\left(2, 11\right)$\label{2f}]{%
        \begin{tikzpicture}[ scale = 0.68 ]
	\begin{semilogyaxis}[xlabel=Crossover probability, ylabel=Block error rate,, ylabel=Block error rate, xmin=0.32, xmax=0.37, ymax=1e-1, ymin=1e-4, grid=both, yminorgrids=true,legend pos=north west, legend style={font=\footnotesize}, legend cell align=left,xticklabels={0.32, 0.33, 0.34, 0.35, 0.36}, xtick={0.32, 0.33, 0.34, 0.35, 0.36}, scaled x ticks=false]
\addplot[solid, thick, mark=x, mark options={solid}, red] coordinates {
(0.3700,3.269281e-01)
(0.3650,2.196922e-01)
(0.3600,1.217381e-01)
(0.3550,6.117909e-02)
(0.3500,2.768498e-02)
(0.3450,1.014053e-02)
(0.3400,3.657697e-03)
(0.3350,1.113682e-03)
(0.3300,2.656988e-04)
(0.3250,5.962802e-05)
};
\addlegendentry{RPL $L = 1024$}
\addplot[solid, thick, mark=star, mark options={solid}, violet] coordinates {
(0.4200,1.000000e+00)
(0.4150,1.000000e+00)
(0.4100,9.996875e-01)
(0.4050,9.993750e-01)
(0.4000,9.974603e-01)
(0.3950,9.916667e-01)
(0.3900,9.685938e-01)
(0.3850,9.190625e-01)
(0.3800,8.295313e-01)
(0.3750,6.955319e-01)
(0.3700,5.198214e-01)
(0.3650,3.387500e-01)
(0.3600,1.923438e-01)
(0.3550,8.777778e-02)
(0.3500,3.409396e-02)
(0.3450,1.081370e-02)
(0.3400,3.008982e-03)
(0.3350,7.012623e-04)
(0.3300,9.987017e-05)
};
\addlegendentry{RPA}
\addplot[solid, thick, mark=+, mark options={solid}, blue] coordinates {
(0.3800,2.765000e-01)
(0.3750,1.669194e-01)
(0.3700,9.253226e-02)
(0.3650,4.629032e-02)
(0.3600,2.009677e-02)
(0.3550,8.393443e-03)
(0.3500,2.945055e-03)
(0.3450,7.949640e-04)
(0.3400,2.155172e-04)
(0.3350,5.831409e-05)
};
\addlegendentry{GS $N = 80$}
\addplot[dashed, thick, mark=square, mark options={solid}, orange] coordinates {
(0.4200,9.845000e-01)
(0.4150,9.565000e-01)
(0.4100,8.920000e-01)
(0.4050,7.880000e-01)
(0.4000,6.360000e-01)
(0.3950,4.632500e-01)
(0.3900,2.845000e-01)
(0.3850,1.556250e-01)
(0.3800,7.052632e-02)
(0.3750,2.578846e-02)
(0.3700,8.462185e-03)
(0.3650,2.208791e-03)
(0.3600,5.567867e-04)
(0.3550,9.960159e-05)
};
\addlegendentry{GS $N = 1024$}
\addplot[dotted, thick, mark=o, mark options={solid}, black] coordinates {
(0.4200,9.830000e-01)
(0.4150,9.532500e-01)
(0.4100,8.887500e-01)
(0.4050,7.782500e-01)
(0.4000,6.210000e-01)
(0.3950,4.485000e-01)
(0.3900,2.705000e-01)
(0.3850,1.455000e-01)
(0.3800,6.526316e-02)
(0.3750,2.323077e-02)
(0.3700,7.369748e-03)
(0.3650,1.901099e-03)
(0.3600,4.515235e-04)
(0.3550,8.764940e-05)
};
\addlegendentry{ML lower bound}

	\end{semilogyaxis}
	\end{tikzpicture}
        }

\caption{The block error rate performance of RM codes on a BSC. For the recursive permutation list algorithm, we use $L$ to denote the list size. For the graph search algorithm, we use $N$ to denote the number of iterations.}
\label{bsc}
\end{figure*}
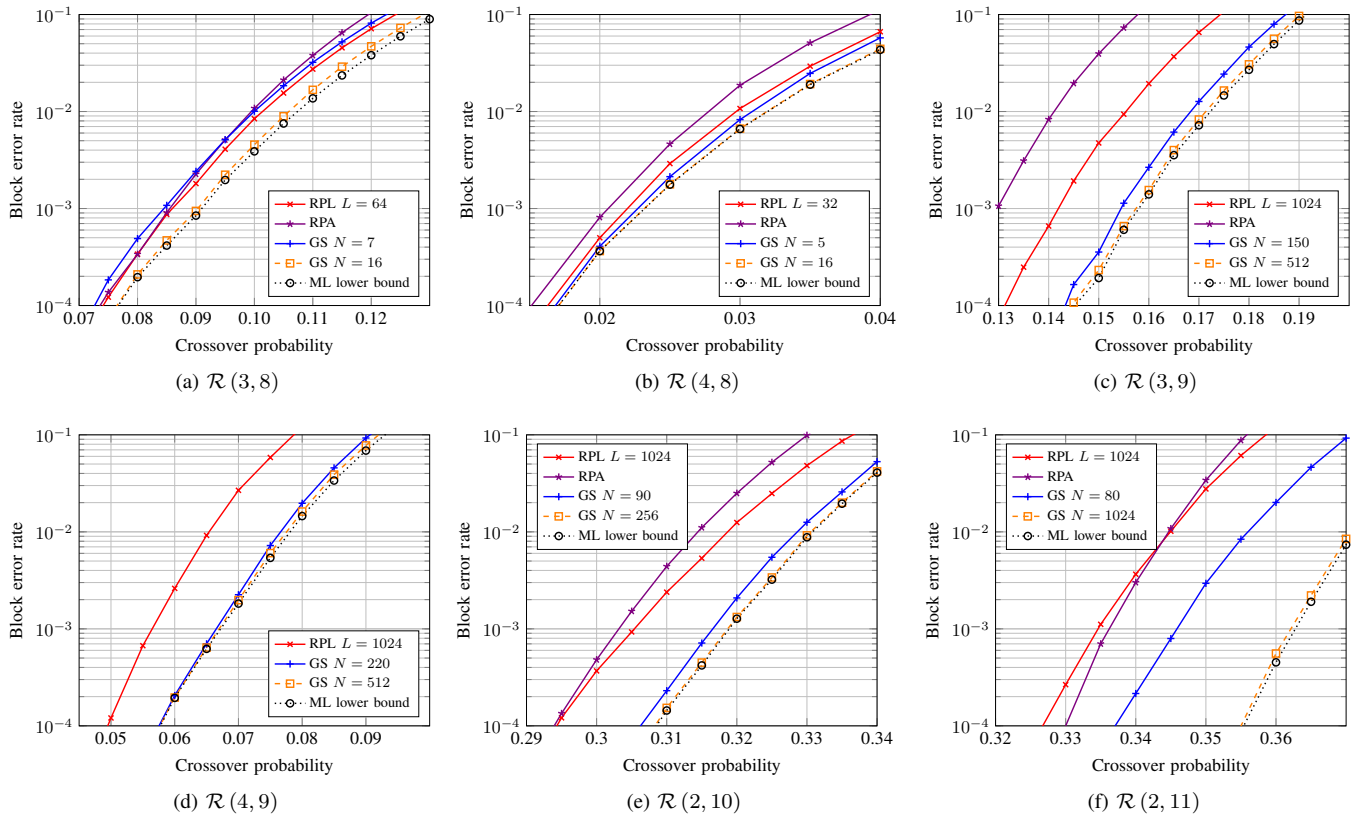

\section{Conclusion}
We presented a new decoder for RM codes, which benefits from the representation of the code as a graph.
Such representation allows using a greedy local search algorithm that is able to find the transmitted codeword efficiently.
In almost all considered cases, our algorithm outperforms the state-of-the-art decoders of RM codes with a similar complexity.
Furthermore, the proposed algorithm allows achieving the performance of the ML decoder with a reasonable worst-case running time on a BSC.
In the case of a BI-AWGN channel, we demonstrated that the ML decoder performance is achieved with a feasible average computational complexity, which can be further reduced using a CRC.

\section*{Acknowledgment}
We thank the reviewers for valuable suggestions that have helped to significantly improve the quality of this paper. 
We also thank Vladimir Gritsenko and Alexey Maevskiy for useful discussions and feedback. 
\bibliographystyle{IEEEtran}
\bibliography{IEEEabrv,myBib}

\end{document}